\documentclass[9pt,a4paper]{amsart} 
\usepackage[margin=3.65cm]{geometry}
\usepackage{amssymb}
\usepackage{graphicx} 
\usepackage[mathscr]{euscript}
\usepackage{enumerate}
\usepackage{xspace}
\usepackage{color}
\usepackage{enumerate}
\usepackage{enumitem}
\usepackage{amsthm}
\usepackage{dsfont}
\usepackage{bbm}
\usepackage{hyperref}
\usepackage[numbers]{natbib}
\DeclareMathAlphabet{\mathpzc}{OT1}{pzc}{m}{it}

\begin{document}

\newcommand\myrot[1]{\mathrel{\rotatebox[origin=c]{#1}{$\Longrightarrow$}}}
\newcommand\NEarrow{\myrot{45}}
\newcommand\SEarrow{\myrot{-45}}

\theoremstyle{plain}
\newtheorem{theorem}{Theorem}[section]
\newtheorem{lemma}[theorem]{Lemma}
\newtheorem{proposition}[theorem]{Proposition}
\newtheorem{corollary}[theorem]{Corollary}
\newtheorem{definition}[theorem]{Definition}
\newtheorem{remark}[theorem]{Remark}
\newtheorem{example}[theorem]{Example}

\newcommand{\lebesgue}{\ensuremath{\lambda\!\!\lambda}}
\renewenvironment{proof}{{\parindent 0pt \it{ Proof:}}}{\mbox{}\hfill\mbox{$\Box\hspace{-0.5mm}$}\vskip 16pt}
\newenvironment{proofthm}[1]{{\parindent 0pt \it Proof of Theorem #1:}}{\mbox{}\hfill\mbox{$\Box\hspace{-0.5mm}$}\vskip 16pt}
\newenvironment{prooflemma}[1]{{\parindent 0pt \it Proof of Lemma #1:}}{\mbox{}\hfill\mbox{$\Box\hspace{-0.5mm}$}\vskip 16pt}
\newenvironment{proofcor}[1]{{\parindent 0pt \it Proof of Corollary #1:}}{\mbox{}\hfill\mbox{$\Box\hspace{-0.5mm}$}\vskip 16pt}
\newenvironment{proofprop}[1]{{\parindent 0pt \it Proof of Proposition #1:}}{\mbox{}\hfill\mbox{$\Box\hspace{-0.5mm}$}\vskip 16pt}

\newcommand{\pK}{\p^{\mathcal{K}}}
\newcommand{\Law}{\ensuremath{\mathop{\mathrm{Law}}}}
\newcommand{\loc}{{\mathrm{loc}}}
\newcommand{\Log}{\ensuremath{\mathop{\mathcal{L}\mathrm{og}}}}
\newcommand{\Meixner}{\ensuremath{\mathop{\mathrm{Meixner}}}}
\newcommand{\of}{[\hspace{-0.06cm}[}
\newcommand{\gs}{]\hspace{-0.06cm}]}
\newcommand{\A}{\mathbf{A}}
\newcommand{\B}{\mathbf{B}}
\newcommand{\p}{\mathds{P}}
\newcommand{\Q}{\mathds{Q}}
\let\MID\mid
\renewcommand{\mid}{|}

\let\SETMINUS\setminus
\renewcommand{\setminus}{\backslash}

\def\stackrelboth#1#2#3{\mathrel{\mathop{#2}\limits^{#1}_{#3}}}

\renewcommand{\theequation}{\thesection.\arabic{equation}}
\numberwithin{equation}{section}

\newcommand\llambda{{\mathchoice
      {\lambda\mkern-4.5mu{\raisebox{.4ex}{\scriptsize$\backslash$}}}
      {\lambda\mkern-4.83mu{\raisebox{.4ex}{\scriptsize$\backslash$}}}
      {\lambda\mkern-4.5mu{\raisebox{.2ex}{\footnotesize$\scriptscriptstyle\backslash$}}}
      {\lambda\mkern-5.0mu{\raisebox{.2ex}{\tiny$\scriptscriptstyle\backslash$}}}}}

\newcommand{\prozess}[1][L]{{\ensuremath{#1=(#1_t)_{0\le t\le T}}}\xspace}
\newcommand{\prazess}[1][L]{{\ensuremath{#1=(#1_t)_{0\le t\le T^*}}}\xspace}
\newcommand{\pH}{\p^{\mathcal{H}}}
\newcommand{\tr}{\operatorname{tr}}
\newcommand{\lijepoa}{{\mathcal{A}}}
\newcommand{\lijepob}{{\mathcal{B}}}
\newcommand{\lijepoc}{{\mathcal{C}}}
\newcommand{\lijepod}{{\mathcal{D}}}
\newcommand{\lijepoe}{{\mathcal{E}}}
\newcommand{\lijepof}{{\mathcal{F}}}
\newcommand{\lijepog}{{\mathcal{G}}}
\newcommand{\lijepok}{{\mathcal{K}}}
\newcommand{\lijepoo}{{\mathcal{O}}}
\newcommand{\lijepop}{{\mathcal{P}}}
\newcommand{\lijepoh}{{\mathcal{H}}}
\newcommand{\lijepom}{{\mathcal{M}}}
\newcommand{\lijepou}{{\mathcal{U}}}
\newcommand{\lijepov}{{\mathcal{V}}}
\newcommand{\lijepoy}{{\mathcal{Y}}}
\newcommand{\cF}{{\mathcal{F}}}
\newcommand{\cG}{{\mathcal{G}}}
\newcommand{\cH}{{\mathcal{H}}}
\newcommand{\cM}{{\mathcal{M}}}
\newcommand{\cD}{{\mathcal{D}}}
\newcommand{\bD}{{\mathbb{D}}}
\newcommand{\bF}{{\mathbb{F}}}
\newcommand{\bG}{{\mathbb{G}}}
\newcommand{\bH}{{\mathbb{H}}}
\newcommand{\dd}{d}
\newcommand{\ddd}{\operatorname{d}}
\newcommand{\er}{{\mathbb{R}}}
\newcommand{\ce}{{\mathbb{C}}}
\newcommand{\erd}{{\mathbb{R}^{d}}}
\newcommand{\en}{{\mathbb{N}}}
\newcommand{\de}{{\mathrm{d}}}
\newcommand{\im}{{\mathrm{i}}}
\newcommand{\set}[1]{\ensuremath{\left\{#1\right\}}}
\newcommand{\indik}{{\mathbf{1}}}
\newcommand{\D}{{\mathbf{D}}}
\newcommand{\E}{{\mathds{E}}}
\newcommand{\N}{{\mathbb{N}}}
\renewcommand{\P}{{\mathbf{P}}}
\newcommand{\ud}{\operatorname{d}\!}
\newcommand{\ii}{\operatorname{i}\kern -0.8pt}
\newcommand{\Var}{\operatorname{Var }\,}
\newcommand{\dt}{\operatorname{d}\!t}   
\newcommand{\ds}{\operatorname{d}\!s}   
\newcommand{\dy}{\operatorname{d}\!y }    
\newcommand{\du}{\operatorname{d}\!u}  
\newcommand{\dv}{\operatorname{d}\!v}   
\newcommand{\dx}{\operatorname{d}\!x}   
\newcommand{\dq}{\operatorname{d}\!q}   
\newcommand{\cadlag}{c\`adl\`ag }

\newcommand{\F}{\mathbf{F}}
\newcommand{\1}{\mathds{1}}
\newcommand{\f}{\mathcal{F}^{\hspace{0.03cm}0}}
\newcommand{\BF}{B^{X, \F}(h)}
\newcommand{\CF}{C^{X, \F}}
\newcommand{\nuF}{\nu^{X, \F}}
\newcommand{\BG}{B^{X, \G}(h)}
\newcommand{\CG}{C^{X, \G}}
\newcommand{\nuG}{\nu^{X, \G}}
\newcommand{\BK}{B^{X, \mathbf{K}}(h)}
\newcommand{\CK}{C^{X, \mathbf{K}}}
\newcommand{\nuK}{\nu^{X, \mathbf{K}}}
\newcommand{\cf}{c^{X, \F}}
\newcommand{\G}{\mathbf{G}}
\newcommand{\M}{\mathcal{M}^{\textup{sp}}}
\newcommand{\K}{\mathbb{K}}
\def\EM{\ensuremath{(\mathbb{EM})}\xspace}
\newcommand{\Cm}{\mathcal{C}_{b, i}^2(\mathbb{R})}
\newcommand{\Cb}{\mathcal{C}_b^2(\mathbb{R})}
\newcommand{\la}{\langle}
\newcommand{\ra}{\rangle}
\newcommand{\uX}{\underline{Z}}
\newcommand{\oX}{\overline{Z}}
\newcommand{\oRd}{|\mathbb{R}^d}
\newcommand{\Norml}[1]{%
{|}\kern-.25ex{|}\kern-.25ex{|}#1{|}\kern-.25ex{|}\kern-.25ex{|}}

\newcommand{\lle}{\langle\hspace{-0.085cm}\langle}
\newcommand{\rre}{\rangle\hspace{-0.085cm}\rangle}

\title[No Arbitrage in Continuous Financial Markets]{No Arbitrage in Continuous Financial Markets}

\author[D. Criens]{David Criens}
\address{D. Criens - Technical University of Munich, Department of Mathematics, Germany}
\email{david.criens@tum.de}
\thanks{D. Criens - Technical University of Munich, Department of Mathematics, Germany,  \texttt{david.criens@tum.de}.}
\thanks{\textbf{Acknowledgment:} The author thanks the referee for excellent reports which helped to improve the article substantially.
	Moreover, he is grateful to the associate editor for many valuable suggestions.}


\date{\today}
\maketitle

\frenchspacing
\pagestyle{myheadings}

\begin{abstract}
We derive integral tests for the existence and absence of arbitrage in a financial market with one risky asset which is either modeled as stochastic exponential of an It\^o process or a positive diffusion with Markov switching.  In particular, we derive conditions for the existence of the minimal martingale measure. We also show that for Markov switching models the minimal martingale measure preserves the independence of the noise and we study how the minimal martingale measure can be modified to change the structure of the switching mechanism.
Our main mathematical tools are new criteria for the martingale and strict local martingale property of certain stochastic exponentials.
\end{abstract}

\section{Introduction}
The absence of arbitrage is of fundamental interest in many areas of financial mathematics. 
Our goal is to provide a systematic discussion for a financial market with one risky asset modeled via its discounted price process \(P = (P_t)_{t \in [0, T]}\), which we assume to be either the stochastic exponential of an It\^o process, i.e. to have dynamics
\begin{align}\label{eq: SEM}
\dd P_t = P_t (b_t \dd t + \sigma_t \dd W_t),
\end{align}
or to be a positive diffusion with Markov switching, i.e. to have dynamics
\begin{align}\label{eq: MSM}
\dd P_t = b(P_t, \xi_t) \dd t + \sigma (P_t, \xi_t) \dd W_t, 
\end{align}
where \(\xi = (\xi_t)_{t \in [0, T]}\) is a continuous-time Markov chain and \(W = (W_t)_{t \in [0, T]}\) is a Brownian motion.

For semimartingale markets the classical concepts of no arbitrage are the notions of \emph{no free lunch with vanishing risk (NFLVR)} as defined by Delbaen and Schachermayer \cite{DelbaenSchachermayer94,DelbaenSchachermayer98} and \emph{no feasible free lunch with vanishing risk (NFFLVR)} as defined by Sin \cite{sin1996strictly}. 
The difference between (NFLVR) and (NFFLVR) is captured by the concept of a \emph{financial bubble} in the sense of Cox and Hobson \cite{Cox2005}.
For our market it is well-known that (NFLVR)  is equivalent to the existence of an \emph{equivalent local martingale measure (ELMM)}, see \cite{DelbaenSchachermayer98}, and that (NFFLVR) is equivalent to the existence of an \emph{equivalent martingale measure (EMM)}, see \cite{Cherny2007,sin1996strictly,Yan97}.
The no arbitrage condition used in the stochastic portfolio theory of Fernholz \cite{fernholz2002stochastic} is \emph{no relative arbitrage (NRA)}. In complete markets Fernholz and Karatzas \cite{fernholz2010} showed that (NRA) is equivalent to the existence of a \emph{strict martingale density (SMD)}.
A weaker concept is \emph{no unbounded profit with bounded risk (NUPBR)}, which is known to be equivalent to the existence of a \emph{strict local martingale density (SLMD)}, see \cite{Choulli1996}. (NUPBR) is considered to be the minimal notion needed for portfolio optimization, see \cite{Karatzas2007}.

The first findings of this article are integral tests for the existence and non-existence of SMDs, ELMMs and EMMs.
For \eqref{eq: SEM} the tests are formulated in terms of Markovian upper and lower bounds for the volatility coefficient \(\sigma\) and for \eqref{eq: MSM} the tests depend on \(x \mapsto \sigma (x, j)\) with \(j\) in the state space of the Markov chain \(\xi\). 
The main novelty of our results is that they apply in the presence of multiple sources of risk.
Beside the Markov switching framework, this is for instance the case in diffusion models with a change point, which represent a change of the economical situation caused for instance by a sudden adjustment in the interest rates or a default of a major financial institution. In general, the question whether (NFLVR) and/or (NFFLVR) hold for a model with a change point is difficult, see \cite{FONTANA20143009} for some results in this direction. Our integral tests provide explicit criteria, which are easy to verify.
For many applications of the Markov switching model \eqref{eq: MSM} it is important to know how the change to an ELMM affects the dynamics of the Markov chain \(\xi\). 
As a second contribution, we study this question form a general perspective for independent sources of risk modeled via martingale problems. In particular, we show that the \emph{minimal local martingale measure (MLMM)}, see \cite{doi:10.1111/j.1467-9965.1992.tb00027.x}, preserves the independence and the laws of the sources of risk. To our knowledge, this property has not been reported in the literature. 
A third contribution of this article are integral tests for the martingale property of certain stochastic exponentials driven by It\^o processes or switching diffusions. These characterizations are our key tools to study the absence of arbitrage. 

We comment on related literature. 
For continuous semimartingale models the absence of arbitrage has been studied by Criens \cite{criens17b}, Delbaen and Shirakawa \cite{Delbaen2002}, Lyasoff \cite{MAFI:MAFI530} and Mijatovi\'c and Urusov \cite{MU(2012)}. Criens, Delbaen and Shirakawa and Mijatovi\'c and Urusov proved integral tests for the existence of SMDs, ELMMs and EMMs in diffusion frameworks. Our results can be viewed as generalizations to an It\^o process or Markov switching framework. For a model comparable to \eqref{eq: SEM}, Lyasoff proved that the existence of an ELMM is determined by the equivalence of a probability measure to the Wiener measure. The structure of this characterization is very different from our results. 
In Section \ref{sec: comments} below we comment in more detail on the results in \cite{criens17b, Delbaen2002,MAFI:MAFI530,MU(2012)}.
The martingale property of stochastic exponentials is under frequent investigation. At this point we mention the articles of Blanchet and Ruf \cite{doi:10.1080/15326349.2015.1114891}, Cheridito et al. \cite{CFY}, Criens \cite{criens17b} and Kallsen and Muhle--Karbe \cite{KMK}. Criens used arguments based on Lyapunov functions and contradictions to verify the martingale property of certain stochastic exponential in a multi-dimensional diffusion setting. We transfer these techniques to a general It\^o process setting. 
Cheridito et al. and Kallsen and Muhle--Karbe related the martingale property of a stochastic exponential to an explosion probability via a method based on the concept of \emph{local uniqueness} as defined in \cite{JS}. This technique traces back to work of Jacod and M\'emin \cite{JM76} and Kabanov et al. \cite{KLS-LACOM1,KLS-LACOM2}.
We use a similar argument for the Markov switching setting. The main difficulties are the proofs of explosion criteria and local uniqueness. 
Both approaches have a close relation to the work of Blanchet and Ruf, where a tightness criterion for the martingale property of non-negative local martingales  has been proven. The connection between Lyapunov functions, explosion and tightness is for instance explained in \cite{doi:10.1080/17442508.2019.1657430}.

Let us also comment on consecutive problems and extensions of our results: In case the discounted price process \(P\) is a positive It\^o process of the type
\[
\dd P_t = b_t \dd t + \sigma_t \dd  W_t,
\] 
our results on the martingale property of stochastic exponentials can be used to obtain characterizations for no arbitrage with a similar structure as for the model \eqref{eq: MSM}.
Moreover, in case \(P\) is the stochastic exponential of a diffusion with Markovian switching, i.e. 
\begin{align*}
\dd P_t &= P_t \dd  S_t, \\
\dd S_t &= b(S_t, \xi_t) \dd t  + \sigma (S_t, \xi_t) \dd W_t,
\end{align*}
our martingale criteria yield conditions for no arbitrage with a similar structure as for \eqref{eq: SEM}.
It is also interesting to ask about multi-dimensional models. In this case, results in the spirit of \cite{criens17b} can be proven by similar arguments as used in this  article. However, the conditions are rather complicated to formulate and space consuming. Therefore, we restrict ourselves to the one-dimensional case.

The article is structured as follows. In Section \ref{sec: MP SE} we give conditions for the martingale and strict local martingale property of certain stochastic exponentials.
In Section \ref{sec: Arb GM} we study the model \eqref{eq: SEM} and in Section \ref{sec: MSM} we study the model \eqref{eq: MSM}. In Section \ref{sec: modifying MLMM} we show that the MLMM preserves independence and laws for sources of risk and we explain how the MLMM can be modified to affect the law of an additional source of risk. The proofs are collected in the remaining sections.

\section{Martingale property of  Stochastic Exponentials}\label{sec: MP SE}
Fix a finite time horizon \(0 < T < \infty\) 
and let \((\Omega, \mathcal{F}, \F, \p)\) be a complete filtered probability space with right-continuous and complete filtration \(\F = (\mathcal{F}_t)_{t \in [0, T]}\). Moreover, fix a state space \(I \triangleq (l, r)\) with  \(- \infty \leq l < r \leq + \infty\).

In the following two sections we provide conditions for the martingale and strict local martingale property of certain stochastic exponentials.
\subsection{The general case}\label{sec: GC}
Assume that \(S = (S_t)_{t \in [0, T]}\) is an \(I\)-valued It\^o process with deterministic initial value \(S_0 \in I\) and dynamics
\[
\dd S_t = b_t \dd t + \sigma_t \dd W_t, 
\]
where \(W = (W_t)_{t \in [0, T]}\) is a one-dimensional Brownian motion and \(b = (b_t)_{t \in [0, T]}\) and \(\sigma = (\sigma_t)_{t \in [0, T]}\) are real-valued progressively measurable processes. It is implicit that \(b\) and \(\sigma\) are such that the integrals are well-defined, i.e. a.s.
\[
\int_0^T \big(|b_s| + \sigma^2_s \big)\dd s < \infty. 
\]

We assume that \(\llambda  \otimes \p\)-a.e.
\(
\sigma \not = 0,
\)
which latter will correspond to the assumption that we consider an asset price process with a non-vanishing volatility.

Let \(c = (c_t)_{t \in [0, T]}\) be a real-valued progressively measurable process such that a.s. 
\[
\int_0^Tc^2_s \dd s < \infty, 
\]
and let \(N = (N_t)_{t \in [0, T]}\) be a local martingale such that a.s. \(\Delta N \geq - 1\) and \([N, W] = 0\).
We ask  for conditions under which the non-negative local martingale
\begin{align}\label{eq:Z}
Z \triangleq \mathcal{E} \Big(N +  \int_0^\cdot c_s \dd W_s \Big),
\end{align}
is a true or a strict local martingale. 
Here, \(\mathcal{E}\) denotes the stochastic exponential. 
The structure of \(Z\) is very important in mathematical finance, because \(Z\) is the prototype of a strict local martingale density, see Lemma \ref{lem: decom} below.

Let \(\underline{a}, \overline{a} \colon I \to (0, \infty), \underline{u}, \overline{u} \colon I \to \mathbb{R}\) and \(\zeta \colon [0, T] \to \mathbb{R}_+\) be Borel functions such that \[\frac{1}{\overline{a}} + \frac{1}{\underline{a}} + |\overline{u}| + |\underline{u}| \in L^1_\textup{loc} (I), \qquad \zeta \in L^1([0, T]).\]
In case \((f, g)\) is one of the pairs \((\underline{u}, \underline{a}), (\underline{u}, \overline{a}), \dots\) we set
\begin{align}\label{eq: v1}
v(f, g)(x) \triangleq \int_{x_0}^x \exp \Big( - \int_{x_0}^y 2 f(z) \dd z \Big) \int_{x_0}^y \frac{2 \exp \big(\int_{x_0}^u2 f(z) \dd z\big)}{g(u)} \dd u \dd y,\quad x \in I,
\end{align}
where \(x_0 \in I\) is fixed.
Let \(l_n \searrow l, r_n \nearrow r\) be sequences such that \(l < l_{n+1} < l_n < r_n < r_{n +1} <  r\). 

The first main result of this section is the following:
\begin{theorem}\label{theo: mart Ito}
		Assume the following:
		\begin{enumerate}
\item[\textup{(M1)}] The sequence		\[
		\tau_n \triangleq \inf (t \in [0, T] \colon S_t \not\in (l_n, r_n)), \quad n \in \mathbb{N},  
		\]
		is a localizing sequence for \(Z\), i.e. \(Z_{\cdot \wedge \tau_n}\) is a martingale for every \(n \in \mathbb{N}\). We use the convention that \(\inf (\emptyset) \triangleq \infty\).
		\item[\textup{(M2)}] For  \(\llambda \otimes \p\)-a.a. \((t, \omega) \in [0, T] \times \Omega\)
		\begin{align*}
		\sigma^2_t(\omega) &\leq \zeta (t) \overline{a}(S_t(\omega)),\\
		\underline{u}(S_t(\omega)) \sigma_t^2 (\omega) &\leq b_t(\omega) + c_t(\omega) \sigma_t(\omega),\\
		\overline{u}(S_t(\omega)) \sigma_t^2 (\omega) &\geq b_t(\omega) + c_t(\omega) \sigma_t(\omega).
		\end{align*}
		\item[\textup{(M3)}] \(
		\lim_{x \nearrow r} v(\overline{u}, \overline{a})(x) = \lim_{x \searrow l} v(\underline{u}, \overline{a}) (x) =  \infty.
		\)
		\end{enumerate}
		Then, \(Z\) is a martingale.
\end{theorem}
The proof of this theorem is given in Section \ref{sec: pf}. 
\begin{remark}
(M3) is independent of the choice of \(x_0\), see \cite[Problem 5.5.28]{KaraShre}.
\end{remark}
Next, we provide a counterpart to Theorem \ref{theo: mart Ito}.
Let \(\mathscr{H}\) be the set of all Borel functions \(h \colon \mathbb{R}_+ \to \mathbb{R}_+\) which are starting at zero, are strictly increasing and satisfy
\[
\int_0^\varepsilon \frac{\dd z}{h^2(z)} = \infty \text{ for all } \varepsilon > 0,\]
and let \(\mathscr{K}\) be the set of all Borel functions \(\kappa \colon \mathbb{R}_+ \to \mathbb{R}_+\), which are starting at zero, are strictly increasing and concave and satisfy
\[
\int_0^\varepsilon \frac{\dd z}{\kappa (z)} = \infty \text{ for all } \varepsilon > 0.\]
In case \((f, g)\) is one of the pairs \((\underline{u}, \underline{a}), (\underline{u}, \overline{a}), \dots\) we say that \((f, g)\) satisfies the \emph{Yamada--Watanabe (YW) conditions}, if for every \(n \in \mathbb{N}\) there exist \(h_n \in \mathscr{H}\) and \(\kappa_n \in \mathscr{K}\) such that and for all \(x, y \in [l_n, r_n]\) 
\begin{align*}
|g^\frac{1}{2} (x) - g^\frac{1}{2} (y)| &\leq h_n (|x - y|),\\ |g(x) f(x) - g(y) f(y)| &\leq \kappa_n (|x - y|).
\end{align*}
The second main result of this section is the following:
\begin{theorem}\label{theo: general SLM}
	Assume one of the following conditions:
\begin{enumerate}
		\item[\textup{(SL1)}]
	The pair \((\underline{u}, \underline{a})\) satisfies the YW conditions, for  \(\llambda \otimes \p\)-a.a. \((t, \omega) \in [0, T] \times \Omega\)
	\begin{equation}\label{eq: bounderies}\begin{split}
	\underline{a}(S_t(\omega)) &\leq \sigma^2_t(\omega),\\
	\underline{u}(S_t(\omega)) \sigma_t^2 (\omega) &\leq b_t(\omega) + c_t(\omega) \sigma_t(\omega),
	\end{split}\end{equation}
	and \(
	\lim_{x \nearrow r} v(\underline{u}, \underline{a})(x) < \infty.
	\)
	\item[\textup{(SL2)}]
	The pair \((\overline{u}, \underline{a})\) satisfies the YW conditions,
	for  \(\llambda \otimes \p\)-a.a. \((t, \omega) \in [0, T] \times \Omega\)
	\begin{align*}
	\underline{a}(S_t(\omega)) &\leq \sigma^2_t(\omega),\\
	\overline{u}(S_t(\omega)) \sigma_t^2 (\omega) &\geq b_t(\omega) + c_t(\omega) \sigma_t(\omega),
	\end{align*}
	and \(\lim_{x \searrow l} v(\overline{u}, \underline{a})(x) < \infty.\)
\end{enumerate}
	Then, \(Z\) is a strict local martingale.
\end{theorem}
The proof of this theorem is given in Section \ref{sec: pf}. 
In Section \ref{sec: dis mart theo} below we comment on the assumptions of Theorems \ref{theo: mart Ito} and \ref{theo: general SLM} and related literature.
\subsection{Markov switching case}\label{sec: MG MS}
In this section we consider a special case of the setting from Section \ref{sec: GC} and assume that \(S\) is a switching diffusion. 
Before we introduce the setting in detail, we clarify terminology: A process is called a \emph{Feller--Markov chain} if it is a Markov chain which is a Feller process in the sense that the corresponding transition semigroup is a self-map on the space of continuous functions vanishing  at infinity. For conditions implying that a Markov chain is Feller--Markov we refer to \cite{anderson2012continuous}. It is also important to stress that whenever we have fixed a filtration and a Markov chain, we presume that the Markov chain is Markovian for the given filtration.
All non-explained terminology for Markov chains, such as \emph{irreducible, recurrent,} etc., can be found in \cite{norris_1997}. 

We assume that \(S = (S_t)_{t \in [0, T]}\) is an \(I\)-valued It\^o process with deterministic initial value \(S_0 \in I\) and dynamics
\begin{align}\label{eq: SD}
\dd S_t = b(S_t, \xi_t) \dd t + \sigma (S_t, \xi_t) \dd W_t,
\end{align}
where \(W = (W_t)_{t \in [0, T]}\) is a one-dimensional Brownian motion, \(\xi = (\xi_t)_{t \in [0, T]}\)  is a continuous-time irreducible Feller--Markov chain with state space \(J \triangleq \{1, \dots, N\}, 1 \leq N \leq \infty,\) and deterministic initial value \(j_0 \in J\), and \(b \colon I \times J \to \mathbb{R}\) and \(\sigma \colon I \times J \to \mathbb{R}\backslash \{0\}\)  are Borel functions such that
\begin{align} \label{eq: int aspp}
\frac{1 + |b(\cdot, j)|}{\sigma^2(\cdot, j)} \in L^1_\textup{loc}(I) \text{ for all } j \in J.
\end{align}
It is implicit that the integrals in \eqref{eq: SD} are well-defined.
We allow \(N = \infty\) in which case \(J= \mathbb{N}\). A process of the type \eqref{eq: SD} is called a \emph{switching diffusion} and the elements of \(J\)  are called \emph{regimes}.
 
Let \(c \colon I \times J \to \mathbb{R}\) be a Borel function such that 
\begin{align}\label{eq: loc inte assp}
\frac{c(\cdot, j)}{\sigma (\cdot, j)} \in L^2_\textup{loc}(I) \text{ for all } j \in J.
\end{align} 
\begin{lemma}\label{lem: c finite}
Almost surely \(\int_0^T c^2(S_s, \xi_s) \dd s < \infty\).
\end{lemma}
\begin{proof}
Set \(F \triangleq \{\xi_s \colon s \in [0, T]\}\),
\(m \triangleq \min_{s \in [0, T]} S_s\) and \(M \triangleq \max_{s \in [0, T]} S_s\).
Using that \(\xi\) only makes finitely many jumps in the finite time interval \([0, T]\), the occupation times formula for continuous semimartingales and \eqref{eq: loc inte assp}, we obtain a.s.
\begin{align*}
\int_0^T c^2 (S_s, \xi_s) \dd s &= \int_0^T \Big( \frac{c (S_s, \xi_s)}{\sigma (S_s, \xi_s)} \Big)^2 \dd  [S, S]_s \\&\leq \sum_{j \in F} \int_0^T \Big( \frac{c (S_s, j)}{\sigma (S_s, j)} \Big)^2 \dd  [S, S]_s 
\\&= \sum_{j \in F} \int_{m }^{M} \Big( \frac{c (x, j)}{\sigma (x, j)} \Big)^2 2L^S_T(x) \dd x 
\\&\leq  \max_{y \in [m, M]} 2L_T^S(y)  \sum_{j \in F} \int_{m}^{M}  \Big( \frac{c (x, j)}{\sigma (x, j)} \Big)^2  \dd x  < \infty,
\end{align*}
where \(L^S\) denotes the local time of \(S\). The lemma is proven.
\end{proof}

We are interested in the martingale property of the non-negative local martingale
\[
Z \triangleq \mathcal{E} \Big(\int_0^\cdot c(S_s, \xi_s) \dd W_s \Big).
\]
This definition coincides with \eqref{eq: Z} for the choices \(c = c(S, \xi)\) and \(N = 0\).

Before we state the main result of this section, we fix some notation.
Because \(L^2_\textup{loc} (I) \subset L^1_\textup{loc} (I)\), \eqref{eq: int aspp} and \eqref{eq: loc inte assp} imply that
\[
\frac{|b (\cdot, j) + c (\cdot, j) \sigma (\cdot, j)|}{\sigma^2(\cdot, j)} \in L^1_\textup{loc}(I) \text{ for all } j \in J.
\] 
Thus, we can set
\[
v(x, j) \triangleq \int_{x_0}^x \exp \Big( - \int_{x_0}^y  \frac{2(b + c \sigma)(z, j)}{\sigma^2(z, j)}  \dd z \Big) \int_{x_0}^y \frac{2\exp \big(\int_{x_0}^s \frac{2(b + c \sigma)(z, j)}{\sigma^2(z, j)} \dd z\big)}{\sigma^2(s, j)} \dd s \dd y
\]
for \((x, j) \in I \times J\) and a fixed \(x_0 \in I\).

We say that  \emph{\(\sigma\) satisfies the Engelbert--Schmidt (ES) conditions for \(j \in J\)} if one of the following holds:
\begin{enumerate}
\item[(ES1)] For every compact set \(K \subset I\) there are Borel functions \(f \colon K \to [0, \infty]\) and \(h \colon \mathbb{R} \to [0, \infty]\) and a constant \(c > 0\) such that the following properties are satisfied:
\begin{enumerate}
\item[(i)] \(\frac{f}{\sigma^2 (\cdot, j)} \in L^1 (K)\).
\item[(ii)] For every neighborhood \(U\) of the origin
\[
\int_U \frac{\dd y}{h(y)} = \infty.
\]
\item[(iii)] For all \(x, x + y \in K, y \in (- c, c)\)
\[|\sigma (x  + y, j) - \sigma (x, j)|^2 \leq f(x) h(y).\]
\end{enumerate}
\item[(ES2)] For every compact set \(K \subset I\) there are Borel functions \(g \colon K \to \mathbb{R}\) and \(h \colon \mathbb{R} \to [0, \infty]\) and a constant \(c > 0\) such that the following properties are satisfied:
\begin{enumerate}
	\item[(i)] \(g\) is increasing.
	\item[(ii)] For every neighborhood \(U\) of the origin
	\[\int_U \frac{\dd y}{h(y)} = \infty.\]
	\item[(iii)] For all \(x, x + y \in K, y \in (- c, c)\backslash \{0\}\)
	\[|\sigma (x  + y, j) - \sigma (x, j)|^2 \leq h(y) \frac{|g(x + y) - g(x)|}{|y|}.\]
	\item[(iv)] \(\inf_{x \in K} \sigma (x, j) > 0\).
\end{enumerate}
\end{enumerate}

We say that the Markov chain \(\xi\) is recurrent if it is a recurrent Markov chain when extended to the infinite time interval \(\mathbb{R}_+\).

The following theorem gives an almost complete answer to the question when \(Z\) is a true or strict local martingale. A proof is given in Section \ref{sec: pf MS}. 
  \begin{theorem}\phantomsection\label{theo: mart MS}
 	\begin{enumerate}
 		\item[\textup{(i)}] Suppose that \(c\) is  bounded on compact subsets of \(I \times J\), 
 		that \(\sigma\) satisfies the ES conditions for all \(j \in J\) and that 
 		\begin{align}\label{eq: MS M integral test}
 		\lim_{x \nearrow r} v(x, j) = \lim_{x \searrow l} v(x, j) = \infty \text{ for all } j \in J.
 		\end{align}
 		Then, \(Z\) is a martingale.
 		\item[\textup{(ii)}] Assume that \(\xi\) is recurrent and that there exists a \(j \in J\) such that \(\sigma\) satisfies the ES conditions for \(j\)
 		and
 		\begin{align}\label{eq: explosive regime cond}
 		\lim_{x \nearrow r} v(x, j) < \infty \text{ or } \lim_{x \searrow l} v(x, j) < \infty.
 		\end{align}
 		Then, \(Z\) is a strict local martingale.
 	\end{enumerate}
 \end{theorem}
\begin{remark} The proof of Theorem \ref{theo: mart MS} (ii) is based on a contradiction argument. In case \eqref{eq: explosive regime cond} holds and \(Z\) is a martingale there exists an \(I\)-valued switching diffusion with an explosive regime \(j\). The recurrence of \(\xi\) ensures that this switching diffusion reaches the regime \(j\),  which leads to a contradiction. In case the initial regime \(j_0\) is already explosive, more precisely if \(\sigma\) satisfies the  ES conditions for \(j_0\) and \(\lim_{x \nearrow r} v(x, j_0) < \infty\) or \(\lim_{x \searrow l} v(x, j_0) < \infty\), the recurrence of \(\xi\) is not needed.
\end{remark}
Noting that \(\xi\) is recurrent in case \(N< \infty\), we obtain the following: 
\begin{corollary}\label{coro: suff nece MS}
	Suppose that \(c\) is bounded on compact subsets of \(I \times J\), that \(\sigma\) satisfies the ES conditions for all \(j \in J\) and that \(N < \infty\). Then,
		\(Z\) is a martingale if and only if \eqref{eq: MS M integral test} holds.
\end{corollary}
\begin{proof}
If \(N < \infty\), the  recurrence of \(\xi\) follows from \cite[Theorems 1.5.6, 3.4.1]{norris_1997}. Now, the  claim is due to Theorem \ref{theo: mart MS}.
\end{proof}
In financial applications, \(N\) can be interpreted as the number of states of the business cycle and therefore \(N < \infty\) is a reasonable assumption.

\subsection{Comments on related literature}\label{sec: dis mart theo}
The martingale property of non-negative local martingales is under frequent investigation. We mention a few related works: A general semimartingale setting has been considered in \cite{criens2018EJP, J79,  JS} and a diffusion and/or  jump-diffusion setting has been studied in  \cite{CFY,KMK, LS, MU(2012), RufSDE,Sin}.

To the best of our knowledge, for a general It\^o process or Markov switching setting Theorems \ref{theo: mart Ito}, \ref{theo: general SLM} and \ref{theo: mart MS} are the first results which provide integral tests for the martingale property of certain stochastic exponentials.

For the diffusion case 
\[
\dd S_t = b(S_t) \dd t + \sigma (S_t) \dd W_t,
\]
a complete characterization of the martingale property of the non-negative local martingale
\[
Z = \mathcal{E} \Big( \int_0^\cdot c(S_s)\dd W_s \Big)
\]
has been proven in \cite{MU(2012)} under local integrability conditions. We stress that in \cite{MU(2012)} the diffusion \(S\) is allowed to explode, which is a feature not included in our framework.
Provided \(S\) is non-explosive, the main theorem of \cite{MU(2012)} shows that \(Z\) is a martingale if and only if 
\[
\lim_{x \nearrow r} v(u, \sigma^2) = \lim_{x  \searrow l} v(u, \sigma^2) = \infty, 
\]
where
\(
u \triangleq \frac{b + c \sigma}{\sigma^2}
\)
and \(v\) is defined as in \eqref{eq: v1}.
The same condition is implied by either Theorems \ref{theo: mart Ito} and \ref{theo: general SLM}, or Corollary \ref{coro: suff nece MS}. 
For the strict local martingale property we require that \(\sigma\) satisfies the ES conditions, which are not imposed in \cite{MU(2012)}. 

The key idea underlying Theorems \ref{theo: mart Ito}, \ref{theo: general SLM} and \ref{theo: mart MS} is a local change of measure combined with either a Lyapunov-type argument (in case of Theorem \ref{theo: mart Ito}), a comparison with one-dimensional diffusions (in case of Theorem \ref{theo: general SLM}) or a local uniqueness property (in case of Theorem \ref{theo: mart MS}).

The idea of using a local change of measure is not new. It has for instance been used in \cite{CFY,criens17b, criens2018EJP,RufSDE,Sin}. 
The Lyapunov and comparison arguments were inspired by \cite{criens17b}, where a multi-dimensional diffusions has been studied. To use the ideas in our general setting, we prove a new Lyapunov condition for It\^o processes and we transport the comparison arguments from a multi-dimensional diffusion setting to a one-dimensional It\^o process framework, see Section \ref{sec: pf} below.
The idea of relating local uniqueness to the martingale property of a stochastic exponential traces back to \cite{JM76,KLS-LACOM1,KLS-LACOM2}. More recently, the method was used in \cite{CFY, criens17b, KMK, Sin}.
Although the terminology suggests the converse, local uniqueness is a strong version of uniqueness in law. 
In the proof of Theorem \ref{theo: mart MS} we deduce local uniqueness from pathwise uniqueness by a Yamada--Watanabe-type argument. 

\section{On the absence and existence of arbitrage}
Let \(0< T < \infty\) be a finite time horizon and let \((\Omega, \mathcal{F}, \F, \p)\) be a complete filtered probability space with right-continuous and complete filtration \(\F = (\mathcal{F}_t)_{t \in [0, T]}\).
We consider a financial market consisting of one risky asset with discounted price process \(P = (P_t)_{t \in  [0, T]}\), which is assumed to be a positive continuous semimartingale with deterministic initial value.

	Recall the following classical terminology: A probability measure \(\Q\) is called an \emph{equivalent (local) martingale measure (E(L)MM)} if \(\Q \sim \p\) and \(P\) is a (local) \(\Q\)-martingale. A strictly positive local \(\p\)-martingale \(Z = (Z_t)_{t \in [0, T]}\) with \(Z_0 = 1\) is called a \emph{strict (local) martingale density (S(L)MD)} if \(ZP\) is a (local) \(\p\)-martingale.
	
In the following we study existence and non-existence of SMDs, ELMMs and EMMs in case \(P\) is either the stochastic exponential of an It\^o process or a positive switching diffusion. In case \(P\) is a positive It\^o process or the stochastic exponential of a real-valued switching diffusion similar results can be deduced from the martingale criteria in Section \ref{sec: MP SE}.
\subsection{Stochastic exponential model}\label{sec: Arb GM}
Suppose that \(P\) is the stochastic exponential of the real-valued It\^o process \(S = (S_t)_{t \in [0, T]}\) with deterministic initial value \(S_0 \in  \mathbb{R}\) and dynamics
\begin{align}\label{eq: S ito}
\dd S_t = b_t \dd t + \sigma_t \dd W_t, 
\end{align}
where \(W = (W_t)_{t \in [0, T]}\) is a one-dimensional Brownian motion and \(b = (b_t)_{t \in [0, T]}\) and \(\sigma = (\sigma_t)_{t \in [0, T]}\) are real-valued progressively measurable processes such that the stochastic integrals in \eqref{eq: S ito} are well-defined.
We assume that  \(\llambda  \otimes \p\)-a.e.
\(
\sigma \not = 0,
\)
which corresponds to the assumption that \(P\) has a non-vanishing volatility. 

\subsubsection{Absence of arbitrage}
In the following we study when a SMD,  ELMM or EMM exists.
As a minimal condition we assume that (NUPBR) holds. This is equivalent to the existence of a \emph{market price of risk \(\theta = (\theta_t)_{t \in [0, T]}\)}, i.e. a real-valued progressively measurable process such that a.s.
\[
\int_0^T \theta^2_s \dd s < \infty 
\]
and
\begin{align} \label{eq: def MPR}
\llambda \otimes \p\text{-a.e.} b - \theta \sigma = 0. 
\end{align}
We define the continuous local martingale 
\begin{align}\label{eq: Z}
Z \triangleq \mathcal{E} \Big(- \int_0^\cdot  \theta_s \dd W_s \Big).
\end{align}
Integration by parts and \eqref{eq: def MPR} yield that 
\begin{align}\label{eq: prod rule LM}
\dd Z_t P_t 
= Z_t P_t (\sigma_t - \theta_t)\dd W_t,
\end{align}
which shows that \(ZP\) is a local martingale or, equivalently, that \(Z\) is a SLMD. We observe the following:
\begin{enumerate}
	\item[\textup{(O1)}] If \(ZP\) is a martingale, then \(Z\) is a SMD by definition.
	\item[\textup{(O2)}] If \(Z\) is a martingale, we can define a probability measure \(\Q\) by the Radon--Nikodym derivative \(\frac{\dd \Q}{\dd \p} \triangleq Z_T\) and \(\Q\) is an ELMM by \eqref{eq: prod rule LM} and \cite[Proposition III.3.8]{JS}.
	\item[\textup{(O3)}] If \(ZP\) and \(Z\) are martingales, then \(\Q\) as defined in (O2) is an EMM by \cite[Proposition III.3.8]{JS}. 
\end{enumerate} 
In summary, to prove the existence of a SMD, ELMM and EMM we have to identify conditions for the martingale property of \(ZP\) and \(Z\).
The following is the main result of this section:

\begin{theorem}\label{theo: main1 SEM}
Suppose the following:
\begin{enumerate}
\item[\textup{(L1)}] The sequence \begin{align}\label{eq: loc seq L1}
\tau_n \triangleq \inf (t \in [0, T] \colon |S_t| \geq n ), \quad n \in \mathbb{N}, 
\end{align} is a localizing sequence for \(Z\). 
\item[\textup{(L2)}]
	There are Borel functions \(\overline{a} \colon \mathbb{R} \to (0, \infty)\) and \(\zeta \colon [0, T] \to \mathbb{R}_+\) such that 
	\begin{align*}
	\frac{1}{\overline{a}} \in L^1_\textup{loc}(\mathbb{R}), \quad \zeta \in L^1([0, T]),
	\end{align*}
	and \(\sigma^2_t(\omega) \leq \zeta(t) \overline{a} (S_t (\omega))\) for \(\llambda \otimes \p\)-a.a. \((t, \omega) \in [0, T] \times \Omega\).
\end{enumerate}
Then, \(Z\) is a martingale, \(\Q\) defined by \(\frac{\dd \Q}{\dd \p} \triangleq Z_T\) is an ELMM and 
	\begin{align}\label{eq:def B} 
	B = W + \int_0^\cdot \theta_t \dd t
	\end{align}
	is a \(\Q\)-Brownian motion such that 
	\[
	S = S_0 + \int_0^\cdot \sigma_t \dd B_t.
	\]
	If in addition
	\begin{align} \label{eq: cond SMD}
	\int_1^\infty \frac{\dd z}{\overline{a}(z)}  = \infty, 
	\end{align}
	then \(\Q\) is an EMM and \(Z\) is a SMD.
\end{theorem}
\begin{proof}
	We apply Theorem \ref{theo: mart Ito} with \(I \triangleq \mathbb{R}, l_n \triangleq - n, r_n \triangleq n\) and \(c \triangleq - \theta\). Note that (L1) equals (M1). Furthermore, set \(\underline{u} (x) \equiv \overline{u} (x) \triangleq 0\). Then, (L2) implies (M2), because \eqref{eq: def MPR} implies \(\llambda \otimes \p\)-a.e. \(b + c \sigma = 0\). Finally, note that 
	\[
	\int_{x_0}^{\pm \infty} \exp \Big(- 2 \int_{x_0}^x \underline{u} (y) \dd y \Big) \dd  x = \int_{x_0}^{\pm \infty} \exp \Big(- 2 \int_{x_0}^x \overline{u} (y) \dd y \Big) \dd  x = \pm \infty,
	\]
	which shows that (M3) holds due to \cite[Problem 5.5.27]{KaraShre}. We conclude that \(Z\) is a martingale and that \(\Q\) is an ELMM by (O2).
	
	Next, we assume that \eqref{eq: cond SMD} holds. We apply Theorem \ref{theo: mart Ito} with \(I \triangleq \mathbb{R}, l_n \triangleq -n, r_n \triangleq n\) and \(c \triangleq \sigma - \theta\) to show that the  local martingale
	\[
	Z' \triangleq \frac{Z P}{P_0} = \mathcal{E} \Big( \int_0^\cdot (\sigma_s - \theta_s) \dd W_s  \Big)
	\]
	is a martingale. In this case, \(\Q\) is an EMM and \(Z\) is a SMD by (O1) and (O3).
	By (L1), the set \(\{Z_{\gamma\wedge \tau_n} \colon \gamma \ \textup{stopping time}\}\) is uniformly integrable (see \cite[Proposition I.1.47]{JS}). Thus,  
	\begin{align*}
	\sup_{\gamma} \E^\p \big[ &Z'_{\gamma \wedge \tau_n}  \1_{\{Z'_{\gamma \wedge \tau_n} \geq K\}} \big] \\&\leq e^{|S_0| + n} \sup_{\gamma} \E^\p \big[ Z_{\gamma \wedge \tau_n} \1_{\{Z_{\gamma \wedge \tau_n} \geq e^{- |S_0| - n}K\}} \big] \to 0 \text{ as } K \to \infty, 
	\end{align*}
	where the \(\sup_{\gamma}\) is meant to be the supremum over all stopping times \(\gamma\).
	Due to \cite[Proposition I.1.47]{JS}, we conclude that (M1) holds for \(Z'\).
	Note that \eqref{eq: def MPR} implies that \(\llambda \otimes \p\)-a.e. \(b + c \sigma = \sigma^2\). Thus, we set \(\underline{u} (x)\equiv \overline{u}(x) \triangleq  1\) and note that (L2) implies (M2) for \(Z'\). Using Fubini's theorem and  \eqref{eq: cond SMD}, we obtain that
	\begin{align*}
	\lim_{x \nearrow \infty} v (1, \overline{a}) (x) &= 2 \int_{x_0}^\infty e^{- 2 y} \int_{x_0}^y \frac{e^{2u}}{\overline{a}(u)} \dd u \dd y
	\\&= 2 \int_{x_0}^\infty \frac{e^{2u}}{\overline{a}(u)}  \int_u^\infty e^{- 2y} \dd y \dd u
	\\&= \int_{x_0}^\infty \frac{\dd u}{\overline{a}(u)} = \infty.
	\end{align*}
	Because
	\[
	\int_{x_0}^{- \infty} \exp \Big( - 2 \int_{x_0}^x \dd y \Big) \dd x = - \infty, 
	\]
	\cite[Problem 5.5.27]{KaraShre} yields that 
	 \(\lim_{x \searrow - \infty} v(1, \overline{a}) (x) = \infty\). Hence, (M3) holds for \(Z'\). We conclude that \(Z'\) is a martingale and the proof is complete.
\end{proof}
In our setting there might exist several ELMMs and it is an important question which ELMM should be chosen for applications. The ELMM from Theorem \ref{theo: main1 SEM} is the \emph{minimal local martingale measure (MLMM)} as defined in \cite{doi:10.1111/j.1467-9965.1992.tb00027.x}.\footnote{In \cite{doi:10.1111/j.1467-9965.1992.tb00027.x} the MLMM has been called \emph{minimal martingale measure}. Because we distinguish between ELMMs and EMMs we adjust the terminology.}  For financial interpretations of the MLMM we refer to \cite{doi:10.1111/j.1467-9965.1992.tb00027.x} and for a general overview on possible applications we refer to \cite{FS2010}.
In Theorem \ref{theo: indp preserving} below we discover a new property of the MLMM: The MLMM preserves independence and laws of sources of risk.

In the following paragraph we relate the assumptions (L1) and (L2) to so-called \emph{weakly equivalent local martingale measures (WELMM)} as introduced in \cite{Kardaras2010}. 
We explain the connection from a general point of view under the assumptions that \(\mathcal{F} = \mathcal{F}_T\) and that (NUPBR) holds.
With slight abuse of notation, let \(Z = (Z_t)_{t \in [0, T]}\) be a SLMD with localizing sequence \((\tau_n)_{n \in \mathbb{N}}\). For every \(n \in \mathbb{N}\) we can define a probability measure \(\Q^n\) by the Radon--Nikodym derivative \(\frac{\dd \Q^n}{\dd \p} \triangleq Z_{T \wedge \tau_n}\). It is easy to see that \(\Q^n\) is an ELMM for the stopped process \(P_{\cdot \wedge \tau_n}\). In other words, for every \(n \in \mathbb{N}\) the notion (NFLVR) holds for all admissible strategies which invest riskless after \(\tau_n\).
Roughly speaking, this observation suggests that (NFLVR) holds in case we can take the limit \(n \to \infty\).
As explained in Section 2.4.2 of \cite{Kardaras2010}, Alaoglu's theorem yields that \((\Q^n)_{n \in \mathbb{N}}\) has an accumulation point \(\mathsf{Q}\) for the weak\(^*\) topology on the dual of \(L^\infty (\Omega, \mathcal{F}, \p)\), which is a finitely additive probability such that \(\mathsf{Q}(A) = 0\) for all \(A \in \mathcal{F}\) with \(\p(A) = 0\), see the Appendix of \cite{Cvitanic2001}. We use the sans-serif typeface to highlight that \(\mathsf{Q}\) is not necessarily a probability measure, because it may fail to be countably additive. Note that \(\mathsf{Q} = \Q^n\) on \(\mathcal{F}_{\tau_n}\) for every \(n \in \mathbb{N}\).
Using this fact, it follows that for all \(A \in \mathcal{F}\) with \(\mathsf{Q}(A) = 0\) we also have \(\p(A) = 0\), 
which shows that \(\mathsf{Q}\) and \(\p\) have the same null-sets. Indeed, if \(A \in \mathcal{F} = \mathcal{F}_T\) is such that \(\mathsf{Q}(A) = 0\), we have \(A \cap \{\tau_n> T\} \in \mathcal{F}_{\tau_n}\) and consequently
\[
\Q^n(A \cap \{\tau_n> T\}) = \mathsf{Q} (A \cap \{\tau_n > T\}) = 0
\]
for all \(n \in \mathbb{N}\).
This implies \(\p(A \cap \{\tau_n > T\}) = 0\) and, because \(\p\)-a.s. \(\tau_n \nearrow \infty\) as \(n \to \infty\), we conclude that \(\p(A) = 0\). 
Following \cite{Kardaras2010}, we call \(\mathsf{Q}\) a WELMM. 
The main difference between WELMMs and ELMMs, and therefore between (NUPBR) and (NFLVR), is that a WELMM is not necessarily a measure.

The idea of condition (L1) is to identify a WELMM, which, as explained above, is a natural candidate for an ELMM. Assuming that \((\tau_n)_{n \in \mathbb{N}}\) is given by \eqref{eq: loc seq L1} means controlling the MPR via the size of the asset. 
This assumption is reasonable from a modeling perspective, because, as explained by Lyasoff \cite[p. 488]{MAFI:MAFI530},
"excessively large expected instantaneous net returns from risky securities entail  excessively large demands for money (to invest in such securities), which, in turn, means higher and higher interest rates, which, in turn, means lower and lower market price of risk". 
In the diffusion settings of Mijatovi\'c and Urusov \cite{Mijatovic2012}, (L1) is equivalent to the local integrability condition \cite[Eq. 3.2]{Mijatovic2012} on the MPR, see \cite[Lemma 6.3]{MU(2012)}. 

Condition (L2) takes care on the countable additivity of the candidate WELMM, which corresponds to problems arising when \(n \to \infty\). Indeed, \(\mathsf{Q}\) is countably additive if and only if
\begin{align}\label{eq: to show} 
\limsup_{n \to \infty} \mathsf{Q} (\tau_n > T) = \limsup_{n \to \infty}\Q^n (\tau_n > T) = 1,
\end{align}
which is also the condition we check in the proof of Theorem \ref{theo: mart Ito}.
If \(\mathsf{Q}\) is countably additive, then \eqref{eq: to show} follows from the monotone convergence theorem and the fact that \(\p\)-a.s. \(\tau_n \nearrow \infty\) as \(n \to \infty\). Conversely, assume that \eqref{eq: to show} holds. Let \((E_k)_{k \in \mathbb{N}} \subset \mathcal{F}\) be a decreasing sequence with \(\bigcap_{k \in \mathbb{N}}  E_k = \emptyset\). Then, because \(E_k \in \mathcal{F} = \mathcal{F}_T\), we have \(E_k \cap \{\tau_n > T\} \in \mathcal{F}_{\tau_n}\),  which yields that 
\begin{align*}
\limsup_{k \to \infty} \mathsf{Q}(E_k) &\leq \mathsf{Q} (\tau_n \leq T) + \limsup_{k \to \infty} \mathsf{Q}  (E_k \cap \{\tau_n > T\})
\\&= \mathsf{Q} (\tau_n \leq T) + \limsup_{k \to \infty} \Q^n  (E_k \cap \{\tau_n > T\}) 
\\&= \mathsf{Q}  (\tau_n \leq T) \to 0 \text{ with } n \to \infty.
\end{align*}
Thus, \(\mathsf{Q}\) is continuous at zero,  which implies that it is countably additive.

\subsubsection{Existence of a financial bubble}
In Theorem  \ref{theo: main1 SEM} we gave conditions for the existence of an ELMM. In this section, we derive a counterpart to \eqref{eq: cond SMD}, which implies the existence of a financial bubble in the sense of \cite{Cox2005}.

As we explain next, the question when a SMD exists is strongly connected to the question when a non-negative local martingale is a strict local martingale. We recall the following: 
\begin{lemma}\label{lem: decom}
If \(Z\) is a SLMD, then there exists a market price of risk \(\theta = (\theta_t)_{t \in [0, T]}\) and a local martingale \(N = (N_t)_{t \in [0, T]}\) such that a.s. \(\Delta N > -1, [N, W] = 0\) and 
\begin{align}\label{eq: decom Z}
Z = \mathcal{E} \Big(N - \int_0^\cdot \theta_s \dd W_s \Big).
\end{align}
\end{lemma}
\begin{proof}
	See \cite[Theorem 1]{doi:10.1080/07362999508809418}.
\end{proof}
In case \(Z\) is a SMD, \eqref{eq: decom Z} holds and 
\begin{align}\label{eq: ZP}
Z P = P_0 \mathcal{E} \Big( N + \int_0^\cdot (\sigma_s-\theta_s) \dd W_s \Big)
\end{align}
is a martingale by definition. If this is not the case, we have a contradiction and no  SMD exists.

The following is the main result of this section:
\begin{theorem}\label{theo: no SMD gen}
	Suppose there exists a Borel function \(\underline{a} \colon \mathbb{R}\to (0, \infty)\) such that \((1, \underline{a})\) satisfies the  YW conditions (see Section \ref{sec: GC} for this terminology), \(\underline{a} (S_t (\omega)) \leq \sigma^2_t (\omega)\) for \(\llambda \otimes \p\)-a.a. \((t, \omega) \in [0, T] \times \Omega\) and 
	\begin{align}\label{eq: cond  no SMD}
	\int_1^\infty \frac{\dd z}{\underline{a}(z)} < \infty.
	\end{align}
	Then, no SMD exists. 
\end{theorem}
\begin{proof}
	We use  Theorem \ref{theo: general SLM} with \(I \triangleq \mathbb{R}\) and
	\(\underline{u} \triangleq 1\) to show that \(ZP\) as defined in \eqref{eq: ZP} is a strict local martingale. Because \(\theta\) is a MPR, \(\llambda \otimes \p\)-a.e. \(b + (\sigma - \theta) \sigma = \sigma^2 = \underline{u}(S) \sigma^2\). Furthermore, Fubini's theorem and \eqref{eq: cond  no SMD} yield that
	\[
	\lim_{x \nearrow \infty}v(1, \underline{a}) (x) = 
	 \int_{x_0}^\infty \frac{\dd z}{\underline{a}(z)} < \infty.
	\]
	Thus, the conditions from part (ii) of Theorem \ref{theo: general SLM} hold and we conclude that \(ZP\) is a strict local martingale. Consequently, as explained above, no SMD exists. 
\end{proof}

The conditions \eqref{eq: cond SMD} and \eqref{eq: cond  no SMD} provide a test for the MLMM to be an EMM or not. 
In a diffusion setting the conditions boil down to a single sufficient and necessary condition, which is also given in \cite[Proposition 5.2]{criens17b}.

\subsubsection{Example: Diffusion models with a change point}
	Fontana et al. \cite{FONTANA20143009} study (NUPBR) and (NFLVR) for a model with a change point. The authors are interested in the influence of filtrations, which represent different levels of information. Under a weak form of the \(\mathcal{H}'\)-hypothesis the model can be included into our framework. More precisely, in this case \(S\) is of the form
\[
\dd S_t = \mu_t \dd t + \big(\sigma^{(1)}(t, S_t) \1_{\{t \leq \tau\}} + \sigma^{(2)}(t, S_t) \1_{\{t > \tau\}}\big) \dd W_t, 
\]
where \(\tau\) is a stopping time. The coefficient \(\sigma^{(i)}\) is assumed to be positive, continuous and Lipschitz continuous in the second variable uniformly in the first, see \cite[Condition I]{FONTANA20143009}.
Theorem \ref{theo: main1 SEM}  provides local conditions for (NFLVR). For instance in the special cases described in \cite[Section 3.3]{FONTANA20143009}, Theorem \ref{theo: main1 SEM} yields that (NFLVR) always holds, because 
\begin{align}\label{eq: mu assp Fea}
\mu_t = \mu^{(1)} (t, S_t) \1_{\{t \leq \tau\}} + \mu^{(2)}(t, S_t) \1_{\{t > \tau\}}, 
\end{align}
where \(\mu^{(i)}\) is locally bounded. This extends the observation from \cite{FONTANA20143009} that (NUPBR) holds in these cases.
Furthermore, if in addition to \eqref{eq: mu assp Fea} for \(i = 1, 2\)
\[
\big(\sigma^{(i)} (t, x)\big)^2 \leq \textup{const. } x, \quad (t, x) \in [0, T] \times [1, \infty),
\]
then even (NFFLVR) holds. The notion (NFFLVR) has not been studied in \cite{FONTANA20143009}.
\subsection{Diffusion model with Markov switching}\label{sec: MSM}
In this section, we assume that \(P\) is a positive continuous semimartingale with deterministic initial value \(P_0 \in (0, \infty)\) and dynamics
\[
\dd P_t = b(P_t, \xi_t) \dd t + \sigma (P_t, \xi_t) \dd W_t, 
\]
where \(W = (W_t)_{t \in [0, T]}\) is a one-dimensional Brownian motion, \(\xi = (\xi_t)_{t \in [0, T]}\)  is a continuous-time irreducible Feller--Markov chain with state space \(J \triangleq \{1, \dots, N\}\), \(1 \leq N \leq \infty\), and deterministic initial value \(j_0 \in J\), and \(b \colon (0, \infty) \times J \to \mathbb{R}\) and \(\sigma \colon (0, \infty) \times J \to \mathbb{R}\backslash \{0\}\)  are Borel functions such that 
\begin{align*} 
\frac{1 + |b(\cdot, j)|}{\sigma^2(\cdot, j)} \in L^1_\textup{loc}((0, \infty)) \text{ for all } j \in J.
\end{align*}

We can interpret \(N\) as the number of all possible states of the business cycle. The assumption of irreducibilily means that we exclude all states of the business cycle which are not attainable from the initial state. We assume \(\xi\) to be a Feller process for technical reasons. In case \(N < \infty\) any Markov chain with values in  \(J\) is a Feller process, because all real-valued functions on \(J\) are continuous and vanishing at infinity.
Due to Lemma \ref{lem: indep MC BM} in the Appendix, the sources of risk \(\xi\) and \(W\) are independent. 
The lemma even shows that it is not possible to model \(\xi\) and \(W\) as Markov processes for a superordinate filtration without their independence. This observation gives a novel interpretation for the independence assumption, which is typically interpreted as the price process being influenced by the business cycle and an additional independent source of risk represented by the driving Brownian motion. 

\subsubsection{Absence and existence of arbitrage}\label{sec: existence MM MS}
We impose the following two assumptions: The coefficient \(b\) is bounded on compact subsets of \((0, \infty) \times J\), \(\sigma^2\) is bounded away from zero on compact subsets of \((0, \infty) \times J\) and \(\sigma\) satisfies the ES conditions for all \(j \in J\), see Section \ref{sec: MG MS} for this terminology. 

We define
\[
\theta(x, j) \triangleq \frac{b(x, j)}{\sigma (x, j)}, 
\]
which is a Borel map bounded on compact subsets of \((0, \infty) \times J\). The process \(\theta_t \triangleq \theta (P_t, \xi_t)\) is a MPR. We define the continuous local martingale \(Z\) as in \eqref{eq: Z}. Note that the observations (O1) -- (O3) in Section \ref{sec: Arb GM} also hold in this setting.
	We call the E(L)MM \(\Q\) with Radon--Nikodym derivative \(\frac{\dd \Q}{\dd \p} = Z_T\) the \emph{minimal (local) martingale measure (M(L)MM)}.
The following theorem provides conditions for the existence of the M(L)MM and for \(Z\) to be a SMD.
		\begin{theorem}\phantomsection\label{theo: main existence MS}
\begin{enumerate}
	\item[\textup{(i)}] Assume that
	\begin{align}\label{eq: MS iff}
	\int_0^1 \frac{z}{\sigma^2(z, j)} \dd z= \infty \text{ for all } j \in J.
	\end{align}
	Then, \(Z\) is a martingale and the probability measure \(\Q\) defined by the Radon--Nikodym derivative \(\frac{\dd \Q}{\dd \p} \triangleq Z_T\) is an ELMM. Moreover, \(B\) as defined in \eqref{eq:def B} is a \(\Q\)-Brownian motion such that 
	\[P = P_0 + \int_0^\cdot \sigma (P_t, \xi_t)\dd B_t.
	\]
	If in addition 
	\begin{align}\label{eq: SMD MS}
		\int_1^\infty \frac{z}{\sigma^2(z, j)} \dd z = \infty \text{ for all } j \in J,
	\end{align}
	then \(\Q\) is an EMM.
	\item[\textup{(ii)}] If \eqref{eq: SMD MS} holds, then \(Z\) is a SMD.
\end{enumerate}
	\end{theorem}
\begin{proof}
The claim follows similar to the proof of Theorem \ref{theo: main1 SEM} when Theorem \ref{theo: mart MS} is used instead of Theorem \ref{theo: mart Ito}.
\end{proof}
Theorem \ref{theo: mart MS} suggests that in case \(\xi\) is recurrent, the conditions in Theorem \ref{theo: main existence MS} are sufficient and necessary. The following theorem makes this precise.
\begin{theorem}\label{theo: main non existence MS} Suppose that \(\xi\) is recurrent. 
	\begin{enumerate}
	\item[\textup{(i)}] 
	If there exists a \(j \in J\) such that 
	\begin{align}\label{eq: no MLMM}
	\int_0^1 \frac{z}{\sigma^2(z, j)} \dd z < \infty, 
	\end{align}
	then \(Z\) is a strict local martingale and the MLMM does not exist. 
	\item[\textup{(ii)}]
	If there exists a \(j \in J\) such that
	\begin{align}\label{eq: no MMM}
	\int_1^\infty \frac{z}{\sigma^2 (z, j)} \dd z < \infty, 
	\end{align}
	then \(Z\) is no SMD. In particular, the MMM does not exist.
	\end{enumerate}
\end{theorem}
\begin{proof}
	The claim follows similar to the proof of Theorem \ref{theo: no SMD gen} when  Theorem \ref{theo: mart MS} is used instead of Theorem \ref{theo: general SLM}.
\end{proof}
Recalling that in case \(N < \infty\) the Markov chain \(\xi\) is recurrent, we obtain the following:
\begin{corollary}\label{coro: MS MMM}
		Suppose that \(N < \infty\). Then, the following are equivalent:
	\begin{enumerate}
		\item[\textup{(a)}]
		The MLMM exists if and only if \eqref{eq: MS iff} holds.
		\item[\textup{(b)}]
		The MMM exists if and only if \eqref{eq: MS iff} and \eqref{eq: SMD MS} hold.
			\item[\textup{(c)}]
		\(Z\) is a SMD if and only if \eqref{eq: SMD MS} holds.
	\end{enumerate}
\end{corollary}
With \(N = 1\) we recover \cite[Corollary 3.4, Theorems 3.6 and 3.11]{Mijatovic2012}. Corollary \ref{coro: MS MMM} means that the M(L)MM exists if and only if the M(L)MM exists for all markets with fixed regimes. 
We will see in the next section that in case one of the frozen markets allows arbitrage, it is not possible to find a risk-neutral market in which the business cycle has Markovian dynamics.

\subsubsection{Non-existence of structure preserving ELMMs and EMMs}
Let \(\mathcal{L}_{\textup{sp}}\) the set of all ELMMs \(\Q\) such that \(\xi\) is an irreducible recurrent Feller--Markov chain on \((\Omega, \mathcal{F}, \F, \Q)\) and let \(\mathcal{M}_\textup{sp}\) be the set of all EMMs in \(\mathcal{L}_\textup{sp}\). 
The main result of this section is the following:

\begin{theorem}\phantomsection\label{theo: NE MS}
	\begin{enumerate}
		\item[\textup{(i)}]
		Suppose there exists a \(j \in J\) such that  \eqref{eq: no MLMM} holds and \(\sigma\) satisfies the ES conditions for \(j\).
		Then, \(\mathcal{L}_\textup{sp} = \emptyset\).
		\item[\textup{(ii)}] 
		Suppose there exists a \(j \in J\) such that  \eqref{eq: no MMM} holds and \(\sigma\) satisfies the ES conditions for \(j\).
		Then, \(\mathcal{M}_\textup{sp} = \emptyset\).
	\end{enumerate}
\end{theorem}
\begin{proof}
	The result follows from the contradiction argument used in the proof of Theorem \ref{theo: general SLM}, where Theorem \ref{theo: existence Markov} has to be used instead of Theorem \ref{theo: 1D Feller p2}.
\end{proof}
In Section \ref{sec: modifying MLMM} we show that an equivalent change to the MLMM does not affect the Markov chain \(\xi\). Thus, Theorem \ref{theo: NE MS} generalizes Theorem \ref{theo: main non existence MS}.
\subsubsection{Example: Markov switching CEV model}
We consider a version of the CEV model (see \cite{Cox15}) with Markov switching.
Take \(\beta \colon J \to (0, \infty)\) and assume that 
	\[
	\sigma (x, j) = x^{\beta(j)}, \quad (x, j) \in (0, \infty) \times J.
	\]
Furthermore, assume that \(b \colon (0, \infty) \times J \to \mathbb{R}\) is locally bounded such that 
\begin{align*}
\int_{1}^\infty \int_{1}^y \frac{\exp (- \int_s^y \frac{2b(z, j)}{z^{2 \beta(j)}} \dd z)}{s^{2  \beta(j)}} \dd s \dd y
= \int_0^1 \int_y^1 \frac{\exp (- \int_s^y \frac{2b(z, j)}{z^{2 \beta(j)}} \dd z)}{s^{2  \beta(j)}} \dd s \dd y = \infty
\end{align*} 
for all \(j \in J\).
Then, the discounted asset price process \(P\) exists due to Theorem \ref{theo: existence Markov} below. 
Let \(Z\) be defined as in \eqref{eq: Z} with \(\theta_t = \frac{b(S_t, \xi_t)}{\sigma(S_t, \xi_t)}\).
In case \(N< \infty\), Corollary \ref{coro: MS MMM} shows the following:
	\begin{enumerate}
		\item[\textup{(a)}] The MLMM exists if and only if \(\beta(j) \geq 1\) for all \(j \in J\).
		\item[\textup{(b)}] The MMM exists if and only if \(\beta(j) = 1\) for all \(j \in J\).
		\item[\textup{(c)}] \(Z\) is a SMD if and only if \(\beta (j) \leq 1\) for all \(j \in J\).
	\end{enumerate}
\subsection{Comments on related literature}\label{sec: comments}
For continuous semimartingale markets the existence and non-existence of SMDs, ELMMs and EMMs has been studied in \cite{criens17b,Delbaen2002, MAFI:MAFI530,  MU(2012)}.
We comment on these works in more detail.

In \cite{Delbaen2002,MU(2012)} a one-dimensional diffusion framework has been considered. We discuss the results from \cite{MU(2012)} and refer to \cite[Remark 3.2]{MU(2012)} for comments on the relation between \cite{Delbaen2002} and \cite{MU(2012)}. In \cite{MU(2012)} it is assumed that the price process \(P = (P_t)_{t \in [0, T]}\) is a \([0, \infty)\)-valued diffusion such that 
\[
\dd P_t = b(P_t) \dd t + \sigma (P_t) \dd W_t,  \quad P_0 \in (0, \infty),
\]
where \(b \colon (0, \infty) \to \mathbb{R}\) and \(\sigma \colon (0, \infty) \to \mathbb{R} \backslash \{0\}\) are Borel functions satisfying
\[
\frac{1 + |b|}{\sigma^2} \in L^1_\textup{loc} ((0, \infty)),
\]
see also \cite[Definition 5.5.20]{KaraShre}. In the following we assume that \(P\) cannot explode to zero. In \cite{MU(2012)} the notions (NFLVR) and (NFFLVR) are also studied in case \(P\) can explode to zero and (NFLVR), (NFFLVR) and (NRA) are further studied for the infinite time horizon. For the non-explosive case the results from \cite{MU(2012)} are as follows: 
\begin{enumerate}
	\item[(a)] (NFLVR) \(\Leftrightarrow\) \(\frac{b}{\sigma} \in L^2_\textup{loc} ((0, \infty))\) and \(\int_0^1 \frac{x}{\sigma(x)} \dd x = \infty\), see \cite[Corollary 3.4]{MU(2012)}.
	\item[(b)] (NFFLVR) \(\Leftrightarrow\) \(\frac{b}{\sigma} \in L^2_\textup{loc} ((0, \infty))\) and \(\int_0^1 \frac{x}{\sigma(x)} \dd x = \int_1^\infty \frac{x}{\sigma (x)} \dd x = \infty\), see \cite[Theorem 3.6]{MU(2012)}.
	\item[(c)] If \(\frac{b}{\sigma} \in L^2_\textup{loc}((0, \infty))\), then (NRA) \(\Leftrightarrow\) \(\int_1^\infty \frac{x}{\sigma (x)} \dd x = \infty\), see \cite[Theorem 3.11]{MU(2012)}.
\end{enumerate}
Applying Corollary \ref{eq: SMD MS} with \(N = 1\) shows versions of (a) -- (c) under slightly more restrictive regularity assumptions on \(b\) and \(\sigma\).
The novelty of Corollary \ref{eq: SMD MS} or more generally Theorems \ref{theo: main existence MS} and \ref{theo: main non existence MS} is their scope of application.

A multi-dimensional diffusion setting has been studied in \cite{criens17b}. We explain the one-dimensional version: Assume that the price process \(P = (P_t)_{t \in [0, T]}\) is the stochastic exponential of
\[
\dd S_t = b(S_t) \dd t + \sigma (S_t) \dd W_t, 
\]
where \(b, \sigma \colon \mathbb{R} \to \mathbb{R}\) are locally bounded Borel functions such that \(\sigma^2\) is locally bounded away from zero. In this setting, \cite[Propositions 5.1]{criens17b} shows that (NFLVR) always holds and \cite[Proposition 5.2]{criens17b} implies that (NFFLVR) \(\Leftrightarrow\) (NRA) \(\Leftrightarrow\) \(\int_1^\infty \frac{\dd x}{\sigma^2(x)} = \infty\).
Under slightly different regularity assumptions on \(b\) and \(\sigma\), the same observation follows from Theorems \ref{theo: main1 SEM} and \ref{theo: no SMD gen}.
The novelty of Theorems \ref{theo: main1 SEM} and \ref{theo: no SMD gen} is that no diffusion structure is needed. In particular, the coefficients \(b\) and \(\sigma\) are allowed to depend on the path of \(S\) or several sources of risk.
In \cite{criens17b} the main interest lies in the multi-dimensional setting. 
We stress that it is possible to extend our results to a multi-dimensional framework. The type of condition will be similar as in \cite{criens17b}.

In \cite{MAFI:MAFI530} the price process \(P = (P_t)_{t \in [0, T]}\) is assumed to be the stochastic exponential of 
\[
\dd S_t = - \alpha (t, S, X) \theta_t \dd t + \alpha (t, S, X) \dd W_t, 
\]
where \(X = (X_t)_{t \in [0, T]}\) is a continuous process, \(\alpha\) and \(\theta\) are suitable processes such that the integrals are well-defined and \(\llambda \otimes \p\)-a.e. \(\alpha \not = 0\). 
The process \(X\) is called \emph{information process}. 
This setting is closely related to those from Section \ref{sec: Arb GM}. 
Let \(\mathscr{W}\) be the Wiener measure and let \(\nu\) be the law of 
\(
- \int_0^\cdot \theta_s \dd s + W.
\)
The main result from \cite{MAFI:MAFI530} is the following: If a.s. \(\int_0^T \theta_s^2 \dd s < \infty\), then (NFLVR) \(\Leftrightarrow \mathscr{W} \sim \nu\),
see \cite[Proposition 2.3]{MAFI:MAFI530}. 
This result is very different from ours, which are intended to give easy to verify conditions for a large class of models.
\section{Modifying Minimal Local Martingale Measures}\label{sec: modifying MLMM}
In Section \ref{sec: existence MM MS} we proved conditions for the existence of the minimal (local) martingale measure in a Markov switching framework. We  ask the following consecutive questions:
\begin{enumerate}
	\item[1.] Does the MLMM change the dynamics of the Markov chain?
	\item[2.] Is it possible to modify the MLMM such that the dynamics of the Markov chain are changed in a tractable manner?
\end{enumerate}
In this section we answer these questions from a general perspective under an independence assumption,  which holds in our Markov switching framework. 
\subsection{Martingale problems}
To characterize additional sources of risk in our financial market, we introduce a martingale problem.

Let \(J\) be a Polish space, define \(D(\mathbb{R}_+, J)\) to be the space of all \cadlag functions \(\mathbb{R}_+ \to J\) and \(\mathcal{D}\) to be the \(\sigma\)-field generated by the coordinate process \(X = (X_t)_{t \geq 0}\), i.e. \(X_t(\omega) = \omega(t)\) for \(\omega \in D(\mathbb{R}_+, J)\) and \(t \in \mathbb{R}_+\). 
We equip \(D(\mathbb{R}_+, J)\) with the Skorokhod topology, which renders it into a Polish space.
It is well-known that \(\mathcal{D}\) is the Borel \(\sigma\)-field on \(D(\mathbb{R}_+, J)\). We refer to \cite{EK,JS} for more details.
Let \(\D^o \triangleq (\mathcal{D}^o_t)_{t \geq 0}\) be the filtration induced by \(X\), i.e. \(\mathcal{D}^o_t \triangleq \sigma (X_s, s \in [0, t])\), and let \(\D \triangleq (\mathcal{D}_t)_{t \geq 0}\) be its right-continuous version, i.e. \(\mathcal{D}_t \triangleq\bigcap_{s > t} \mathcal{D}^o_s\) for all \(t \in \mathbb{R}_+\). 

Let \((B_n)_{n \in \mathbb{N}}\) be an increasing sequence of nonempty open sets in \(J\) such that \(\bigcup_{n \in \mathbb{N}}B_n = J\) and define 
\begin{align}\label{eq: rhon}
\rho_n (\omega) \triangleq \inf\big(t \in \mathbb{R}_+ \colon \omega(t) \not \in B_n \textup{ or } \omega (t-) \not \in B_n\big),\quad \omega \in D(\mathbb{R}_+, J), n \in \mathbb{N}.
\end{align}
Due to \cite[Proposition 2.1.5]{EK}, \(\rho_n\) is a \(\D^o\)-stopping time and, due to \cite[Problem 4.27]{EK}, \(\rho_n \nearrow \infty\) as \(n \to \infty\). We will use the sequence \((\rho_n)_{n \in \mathbb{N}}\) as a localizing sequence for test martingales of our martingale problem. We fix this sequence, because for some arguments we need a common localizing sequence consisting of \(\D^o\)-stopping times. 

The input data for our martingale problem is the following:
\begin{enumerate}
\item[(i)] A set \(A \subseteq C(J, \mathbb{R})\), where \(C(J, \mathbb{R})\) denotes the space of continuous functions \(J\to\mathbb{R}\).
\item[(ii)] A map  \(L \colon A \to \mathcal{PM}\) such that for all \(f \in A, t \in \mathbb{R}_+\) and \(\omega \in D (\mathbb{R}_+, J)\)
\[
\int_0^{t} \big|Lf(\omega, s)\big| \dd s  < \infty,
\]
where \(\mathcal{PM}\) denotes the space of all \(\D\)-progressively measurable processes.
\item[(iii)] An initial value \(j_0 \in  J\).
\item[(iv)] A time horizon \(0 < T \leq \infty\).
\end{enumerate}
We use the convention that in case \(T = \infty\) the interval \([0, T]\) is identified with \(\mathbb{R}_+\).
\begin{definition}
	\begin{enumerate}
		\item[\textup{(i)}] Let \((\Omega^o, \mathcal{F}^o, \F^o, \p^o)\) be a filtered probability space with right-continuous filtration \(\F^o = (\mathcal{F}^o_t)_{t \in [0, T]}\), supporting a c\`adl\`ag, adapted, \(J\)-valued process \(\xi = (\xi_t)_{t \in [0, T]}\).
		We say that \(\xi\) is a \emph{solution process to the martingale problem \((A, L, j_0, T)\)}, if for all \(f \in A\) and \(n \in \mathbb{N}\) the process
		\begin{align}\label{eq: pro MP}
		M^{f, n} \triangleq f (\xi_{\cdot \wedge \rho_n (\xi)})  - f(\xi_0) - \int_0^{\cdot \wedge \rho_n (\xi)} L f (\xi, s) \dd s
		\end{align}
		is a martingale, \(\p^o (\xi_0 = j_0) = 1\) and for all \(t \in [0, T]\) there exists a constant \(C = C(f, n, t) > 0\) such that a.s. \(\sup_{s \in  [0, t]} |M^{f, n}_s| \leq C\).
		\item[\textup{(ii)}] We say that the \emph{martingale problem has a solution} if there exists a filtered probability space which supports a solution process. 
		\item[\textup{(iii)}] We say that the martingale problem satisfies \emph{uniqueness} if the laws (seen as Borel probability measures on \(D(\mathbb{R}_+, J)\)) of any two solution processes, possibly defined on different filtered probability spaces, coincide.
		\item[\textup{(iv)}] If for all \(j_0 \in J\) the martingale problem \((A, L, j_0, T)\) has a solution and satisfied uniqueness, we call the martingale problem \((A, L, T)\) \emph{well-posed}.
	\end{enumerate}
\end{definition}
	Martingale problems were introduced by Stroock and Varadhan \cite{SV} in a diffusion setting. Martingale problems for semimartingales were studied in \cite{J79} and Markovian martingale problems with a Polish state space were studied in \cite{EK}. Our definition is unifying in the sense that it deals with non-Markovian processes and a Polish state space. Most of the conditions for existence and uniqueness given in \cite{EK, J79, SV} also apply to our setting.

\begin{example}[Martingale problem for Markov chains]\label{ex: xi2} 
	Suppose that \(J = \{1, \dots, N\}\) with \(1 \leq N \leq \infty\). We equip \(J\) with the discrete topology. Let \(\xi = (\xi_t)_{t \geq 0}\) be a Feller--Markov chain with initial value \(j_0 \in J\) and \(Q\)-matrix \(Q\).
	Due to \cite[Theorem 5]{doi:10.1112/jlms/s2-5.2.267}, the generator \((\mathcal{L}, D(\mathcal{L}))\) of \(\xi\) is given by 	\(
	\mathcal{L} = Q\) and
	\(
	D(\mathcal{L}) = \{f \in C_0(J) \colon Q f \in C_0(J)\},
	\)
	where \(C_0(J)\) denotes the space of all continuous functions \(J \to \mathbb{R}\) which are vanishing at infinity.
	Due to Dynkin's formula (see \cite[Proposition VII.1.6]{RY}) the process \(\xi\) solves the martingale problem \((D(\mathcal{L}), \mathcal{L}, j_0, \infty)\) and, due to \cite[Theorem 3.33]{liggett2010continuous}, the martingale problem satisfies uniqueness.
	
	Conversely, in case \(\xi\) is a solution process to the martingale problem \((\mathcal{L}, D(\mathcal{L}), j_0, \infty)\), where \((\mathcal{L}, D(\mathcal{L}))\) given as above is the generator of a Feller process, \(\xi\) is a Feller--Markov chain with \(Q\)-matrix \(Q\), see \cite[Theorem 3.4.2]{EK} and \cite[Theorem 3.33]{liggett2010continuous}.
\end{example}

\subsection{How to modify the MLMM}
Fix a finite time horizon \(0 < T < \infty\) and let  \((\Omega, \mathcal{F}, \F, \p)\) be a complete filtered probability space with right-continuous and complete filtration \(\F = (\mathcal{F}_t)_{t \in [0, T]}\), which supports a solution process \(\xi = (\xi_t)_{t \in [0, T]}\) to the martingale problem \((A, L, j_0, T)\). Moreover, assume that the martingale problem \((A, L, j_0, T)\) satisfies uniqueness. 
Let \(W = (W_t)_{t \in [0, T]}\) be a one-dimensional Brownian motion such that \(\sigma  (W_t, t \in [0, T])\) and \(\sigma (\xi_t, t \in [0, T])\) are independent. We think of \(W\) and \(\xi\) as two independent sources of risk influencing the market.
The independence assumption is satisfied when \(\xi\) is a Feller--Markov chain, see Lemma \ref{lem: indep MC BM} in the Appendix. 

In the following theorem we find a new property of the MLMM. To wit, we show that the MLMM preserves the independence of the sources of risk and their laws. Because the M(L)MM is often used for pricing, this observation is important for analytical and numerical computations. We prove the following theorem in Section \ref{sec: pf theo modi ind}.
\begin{theorem}\label{theo: indp preserving}
Let \(c = (c_t)_{t \in [0, T]}\) be a real-valued progressively measurable process such that a.s. \[\int_0^T c_s^2 \dd s < \infty\] and define
\begin{align*}
Z &\triangleq \mathcal{E} \Big( \int_0^\cdot c_s \dd W_s \Big), \quad
B \triangleq W  - \int_0^\cdot c_s \dd s.
\end{align*}
Suppose further that \(Z\) is a martingale and that the martingale problem \((A, L, j_0, T)\) satisfies uniqueness. 
Define \(\Q\) by the Radon--Nikodym derivative \(\frac{\dd \Q}{\dd \p} \triangleq Z_T\). 
Then, \(\sigma ( B_t, t \in [0, T])\) and \(\sigma (\xi_t, t \in [0,T])\) are \(\Q\)-independent, \(B\) is a \(\Q\)-Brownian motion and \(\xi\) is a solution process to the martingale problem \((A, L, j_0, T)\) on \((\Omega, \mathcal{F}, \F, \Q)\).
\end{theorem}
Let us outline an important consequence of Theorem \ref{theo: indp preserving}:
If the MLMM exists, then its density is of the same type as \(Z\) in Theorem \ref{theo: indp preserving} and it follows that the joint law of the sources of risk remains unchanged by an equivalent change to the MLMM. 
In particular, in the setting of Section \ref{sec: MG MS} this means that \(\xi\) stays a Markov chain after a change to the MLMM. 

We ask further whether it is possible to modify the MLMM such that the law of \(\xi\) can be affected in a tractable manner. An answer to this question is provided by the next theorem.
A proof can be found in Section \ref{sec: pf mg JS}.

\begin{theorem}\label{theo: cdc}
	Let \(f \in A\) be strictly positive and suppose that the process
	\begin{align}\label{eq: mart to show bounded}
	Z \triangleq \frac{f(\xi)}{f(j_0)} \exp \Big(- \int_0^\cdot \frac{Lf (\xi, s)}{f(\xi_s)} \dd s \Big)
	\end{align}
	is a martingale.
	Set 	\[
	A^* \triangleq \big\{g \in A \colon fg \in A\big\},
	\]
	and 
	\[
	L^* g \triangleq \frac{L (f g) - g L f}{f}.
	\]
	Suppose that 
	for every \(g \in A^*\) and \(n \in \mathbb{N}\) there exists a constant \(C = C(g, n) > 0\) such that a.s. 
	\[
	\sup_{t  \in [0, T]} \Big| g(\xi_{t \wedge \rho_n(\xi)}) - g(\xi_0) - \int_0^{t \wedge \rho_n(\xi)} L^* g(\xi, s) \dd s \Big| \leq C.
	\]
	Define the probability measure \(\Q\) by the Radon--Nikodym derivative 
	\(
	\frac{\dd \Q}{\dd \p} \triangleq Z_T.
	\)
	Then, \(\sigma(\xi_t, t \in [0, T])\) and \(\sigma(W_t, t \in [0, T])\) are \(\Q\)-independent, \(W\) is a \(\Q\)-Brownian motion and \(\xi\) is a solution process for the martingale problem \((A^*, L^*, j_0, T)\) on \((\Omega, \mathcal{F}, \F, \Q)\).
\end{theorem}
\begin{remark}\label{rem: loc mart}
	\begin{enumerate}
	\item[\textup{(i)}]
	For all \(\omega \in D(\mathbb{R}_+, J)\) and \(g \in A^*\)
	\[
	\int_0^T\Big(\Big|\frac{L f(\omega, s)}{f(\omega(s))}\Big| + \big| L^* g(\omega, s)\big| \Big)\dd s <  \infty, 
	\]
	because \(f\) and \(g\)  are continuous and the set \(\{\omega(t) \colon t \in  [0, T]\} \subseteq J\) is relatively compact, see \cite[Problem 16, p. 152]{EK}.
	Consequently, \(Z\) and the martingale problem \((A^*, L^*, j_0, T)\) are well-defined.
		\item[\textup{(ii)}]
	In view of \cite[Corollary 2.3.3]{EK}, the process \eqref{eq: mart to show bounded} is always a local martingale by the definition of the martingale problem. 
	\end{enumerate}
\end{remark}
We explain an application of Theorem \ref{theo: cdc}:
Suppose that the MLMM exists. Then, using the change of measure described in Theorem \ref{theo: cdc}, the MLMM can be changed further such that the law of \(\xi\) gets affected as described in the theorem, while the local martingale property of the price process is preserved. We stress that in this manner the MLMM induces a family of ELMMs, which is often infinite. 
In a Markov switching framework with \(N < \infty\) the following proposition  explains how the \(Q\)-matrix of the driving Feller--Markov chain can be changed.
\begin{proposition}\label{prop: com MC}
	Suppose that \(J = \{1, \dots, N\}\) with \(N < \infty\) and 
	\[
	L f(\omega, s) = Q f (\omega(s)), \quad \omega \in D(\mathbb{R}_+, J), s \in \mathbb{R}_+,
	\]
	for a \(Q\)-matrix \(Q = (q_{ij})_{i, j \in J}\) and \(f \in A \triangleq \mathbb{R}^N\). Let \(f \in (0, \infty)^N\) and \(A^*,L^*\) as in Theorem \ref{theo: cdc}. Then, \(A^* = \mathbb{R}^N\) and 
		\[
	L^* f(\omega, s) = Q^* f (\omega(s)), \quad f \in \mathbb{R}^N, \omega \in D(\mathbb{R}_+, J), s \in \mathbb{R}_+,
	\]
	for \(Q^* = (q^*_{ij})_{i, j \in J}\)  with
	\begin{align*} 
	q^*_{ij} \triangleq
	\begin{cases}
	q_{ij} \frac{f (j)}{f (i)},&i \not = j,\\
	- \sum_{k \not = i} q_{ik} \frac{f(k)}{f(i)},& i = j.
	\end{cases}
	\end{align*}
\end{proposition}
\begin{proof}
	See \cite[Proposition 5.1]{palmowski2002}.
\end{proof}

A useful criterion for the martingale property of \eqref{eq: mart to show bounded} is given by Theorem \ref{prop: mg f} below. We consider it as an extension of results from \cite{CFY, J79, SV}.
In the following \(X = (X_t)_{t \geq 0}\) denotes the coordinate process on \(D(\mathbb{R}_+,J)\).

\begin{definition} A set \(\widetilde{A} \subseteq A\) is called a \emph{determining set} for the martingale problem \((A, L, \infty)\) if for all \(j_0 \in J\) a Borel probability measure \(\mu\) on \(D(\mathbb{R}_+, J)\) is the law of a solution process to the martingale problem \((A, L, j_0, \infty)\) if and only if for all \(f \in \widetilde{A}\) and \(n \in \mathbb{N}\) the process
	\[
	f(X_{\cdot \wedge \rho_n}) - f(X_0) - \int_0^{\cdot \wedge \rho_n} Lf (X, s)\dd s
	\]
	is a \(\mu\)-martingale and \(\mu(X_0 = j_0) = 1\).
\end{definition}
\begin{example}[Determining set for Feller--Markov chains]
	Let \(J, A\) and \(L\) be as in Example \ref{ex: xi2}. Note that
	\[
	G \triangleq \big\{ (f, Qf) \colon f \in A\big\} \subset C_0(J) \times C_0 (J).
	\]
	Because \(C_0(J)\) equipped with the uniform metric is a separable metric space, \(G\) is a separable metric space when equipped with the taxicap uniform metric.
	Hence, we find a countable set \(\widetilde{A} \subseteq A\) such that for each \((f, g) \in G\) there exists a sequence \((f_n)_{n \in \mathbb{N}} \subset \widetilde{A}\) with
	\[
	\|f_n - f\|_\infty + \|Q f_n - g\|_\infty \to 0 \quad \text{ as } n\to \infty.
	\]
	Due to \cite[Proposition 4.3.1]{EK}, \(\widetilde{A}\) is a determining set for the martingale problem \((A, L, \infty)\).
\end{example}
A proof for the following theorem can be found in Section \ref{sec: pf mg JS}.
\begin{theorem}\label{prop: mg f}
	Let \(f, A^*\) and \(L^*\) be as in Theorem \ref{theo: cdc}. Moreover, assume there exists a countable determining set for the martingale problem \((A^*, L^*, \infty)\) 
	and that
	\[
	L^* g (\xi, t) = Kg(\xi_{t}), \quad f \in A^*, t \in \mathbb{R}_+,
	\]
	where \(K\) maps \(A^*\) into the space of Borel functions \(J \to \mathbb{R}\). Finally, assume that the martingale problem \((A^*, L^*, \infty)\) is well-posed and that \((\rho_n (\xi))_{n \in \mathbb{N}}\) is a localizing sequence for the local martingale \eqref{eq: mart to show bounded}, see Remark \ref{rem: loc mart}. Then, the process \eqref{eq: mart to show bounded} is a martingale.
\end{theorem}
Roughly speaking, this theorem shows that in Markovian settings we can modify the law of \(\xi\) whenever the martingale problem \((A^*, L^*, \infty)\) is well-posed.
\begin{remark}
	The existence of a solution to the martingale problem \((A^*, L^*, j_0, T)\) is often necessary for the martingale property of \(Z\), see Theorem \ref{theo: cdc}. 
\end{remark}

\section{Proof of Theorems \ref{theo: mart Ito} and  \ref{theo: general SLM}}\label{sec: pf}
The following section is divided into three parts. In the first part we prove Lyapunov-type conditions for non-explosion of It\^o processes, in the second part we prove non-existence conditions for It\^o processes and in the third part we deduce Theorems \ref{theo: mart Ito} and  \ref{theo: general SLM}.

\subsection{Criteria for non-explosion}
In this  section we pose ourselves into a version of the setting from Section \ref{sec: GC}.
Let \(I = (l, r)\) be as in Section \ref{sec: GC} and  \((\Omega,\mathcal{F})\) be a measurable space which supports three real-valued processes \(S = (S_t)_{t \in [0, T]}, b = (b_t)_{t \in [0,T]}\) and \(\sigma = (\sigma_t)_{t \in [0, T]}\).
For every \(n \in \mathbb{N}\) we fix a probability measure \(\Q^n\) and a right-continuous \(\Q^n\)-complete filtration \(\F^n= (\mathcal{F}^n_t)_{t \in [0, T]}\) on \((\Omega, \mathcal{F})\) such that \(S, b\) and \(\sigma\) are \(\F^n\)-progressively measurable. We set \(\tau_n\) as in Theorem \ref{theo: mart Ito}, i.e. 
\[
\tau_n =  \inf(t \in [0,  T] \colon  S_t  \not \in (l_n, r_n)), 
\]
where \(l_n \searrow l, r_n \nearrow r\) are sequences such that \(l < l_{n+1} < l_n < r_n < r_{n +1} <  r\).
Moreover, suppose that \(\Q^n\)-a.s.
\[
\dd S_{t \wedge \tau_n} = b_t \1_{\{t \leq \tau_n\}} \dd t + \sigma_t \1_{\{t \leq \tau_n\}} \dd W^n_t, \quad S_0 \in I,
\]
where \(W^n = (W^n_t)_{t \in [0, T]}\) is a Brownian motion on \((\Omega, \mathcal{F}, \F^n, \Q^n)\). It is implicit that the integrals are well-defined.
We also assume that 
\begin{align}\label{eq: nondeg n}
 \llambda \otimes \Q^n\text{-a.e. } \sigma \not = 0 \text{ for all } n \in \mathbb{N}
\end{align}
and we fix a Borel function \(\zeta \colon [0, T] \to \mathbb{R}_+\) such that \(\zeta \in L^1([0, T])\).
\subsubsection{A Lyapunov criterion}\label{sec: cond E}
In this section we give a Lyapunov-type condition for 
 \begin{align}\label{eq: cond qn}
\limsup_{n \to \infty} \Q^n(\tau_n = \infty) = 1. \end{align} 
For \(f \in C^1(I, \mathbb{R})\) with locally absolutely continuous derivative, it is well-known that there exists a \(\llambda\)-null set \(N^f \subset I\) such that \(f\) has a second derivative \(f''\) on \(I \backslash N^f\). In this case, we set 
\[
\mathcal{L} f   \triangleq f'(S) b + \tfrac{1}{2} f'' (S) \1_{I \backslash N^f}(S) \sigma^2.
\]

\begin{theorem}\label{theo: NE 1}
	Let \(V \colon I \to (0, \infty)\) be differentiable with locally absolutely continuous derivative such that 
	\begin{align}\label{eq: minimum convergence}
	\limsup_{n \to \infty} V(l_n) \wedge V(r_n) = \infty.
	\end{align}
	Suppose there exists a \(\llambda\)-null set \(N \subset I\) such that 
	\begin{equation}\label{ineq: Lyapunov2}\begin{split}
	\mathcal{L}V (t) (\omega) \1_{I \backslash N} (S_t(\omega)) \leq \zeta(t)&V(S_t(\omega)) \1_{I \backslash N}(S_t(\omega))  \\&\text{ for } \llambda \otimes \Q^n\text{-a.a. } (t, \omega) \in [0, T] \times \Omega, \quad n \in \mathbb{N}.
	\end{split}
	\end{equation}
	Then, \eqref{eq: cond qn} holds.
\end{theorem}
\begin{proof}
	Let \(L^S\) be the local time of the continuous \(\Q^n\)-semimartingale \(S_{\cdot \wedge \tau_n}\). The occupation times formula yields that  \(\Q^n\)-a.s.
	\begin{align*}
	\int_0^{\tau_n \wedge T} \1_N(S_{s}) \sigma^2_s \dd s = 2 \int_{- \infty}^\infty \1_N(x) L^S_{T} (x) \dd x = 0,
	\end{align*}
	which implies that \(\Q^n\)-a.s \(\llambda (\{t \in [0, \tau_n \wedge T] \colon S_{t} \in N\}) = 0\). We will use this fact in the following without further reference.
	
	Set 
	\begin{align*}
	U^n \triangleq \exp \Big(- \int_0^{\cdot \wedge \tau_n} \zeta(s)\dd s\Big) V(S_{\cdot \wedge \tau_n}).
	\end{align*}
	Using a generalized version of It\^o's formula (see \cite[Lemma IV.45.9]{RW2}), 
	we obtain that 
	the process
	\begin{equation*}\begin{split}
	U^n + \int_0^{\cdot \wedge \tau_n} \exp \Big(- \int_0^{s} \zeta(z) \dd z \Big)\big(\zeta(s) V(S_s) - \mathcal{L} V (s) \big)  \dd s 
	\end{split}
	\end{equation*}
	is a local \(\Q^n\)-martingale.
	We deduce from \eqref{ineq: Lyapunov2} and the fact that non-negative local martingales are supermartingales, that \(\Q^n\)-a.s.
	\begin{align*} 
	U^n \leq \Q^n\textup{-supermartingale starting at } U_0 = V(S_0).
	\end{align*}
W.l.o.g. we assume that \(S_0 \in (l_{1}, r_{1})\).
Note that for all \(n \in \mathbb{N}\) we have \(\Q^n\)-a.s. \(S_{\tau_n} \in \{l_n, r_n\}\) on \(\{\tau_n \leq T\}\).
	We conclude that for all \(n \in \mathbb{N}\)
	\begin{align*}
	\Q^n(\tau_n \leq T) \exp \Big(- \int_0^T \zeta(s)\dd s \Big) (V(l_n) \wedge V(r_n))
	&\leq \E^{\Q^n} \big[ U^n_{\tau_n} \1_{\{\tau_n \leq T\}} \big]
	\\&\leq \E^{\Q^n} \big[ U^n_{T}\big] \leq V(S_0).
	\end{align*}
	By \eqref{eq: minimum convergence} there exists a sequence \((n_k)_{k \in \mathbb{N}} \subset \mathbb{N}\) with \(n_k \to \infty\) as \(k \to \infty\) such that 
	\(
	V(l_{n_k}) \wedge V(r_{n_k}) > 0\) for all \(k \in \mathbb{N}
	\)
	and 
	\(
	\lim_{k \to \infty} V(l_{n_k}) \wedge V(r_{n_k}) = \infty.
	\)
	We deduce from
	\[
	0 \leq \Q^{n_k} (\tau_{n_k} \leq T) \leq V(S_0) \exp \Big(\int_0^T \zeta(s)\dd s \Big) \frac{1}{V(l_{n_k}) \wedge V(r_{n_k})}
	\]
	that 
	\[
	\lim_{k \to \infty} \Q^{n_k} (\tau_{n_k} \leq T) = 0.
	\]
	Because \(\{\tau_n \leq T\}^c = \{\tau_n = \infty\}\), we obtain
	\[
	1 =  \lim_{k \to \infty} \Q^{n_k} (\tau_{n_k} = \infty) \leq \limsup_{n \to \infty} \Q^n(\tau_n = \infty) \leq 1, 
	\]
	which implies  \eqref{eq: cond qn}. The proof is complete.
\end{proof}

\subsubsection{An integral test}
Let \(\overline{a} \colon I \to (0, \infty)\) and \(\underline{u}, \overline{u} \colon I \to \mathbb{R}\) be Borel functions such that 
\[
\frac{1}{\overline{a}} + |\underline{u}| + |\overline{u}| \in L^1_\textup{loc}(I).
\]
Recall from Section \ref{sec: MP SE} that in case \((f, g)\) is one of the pairs \((\underline{u}, \overline{a}), (\overline{u}, \overline{a})\) we set
\begin{align}\label{eq: v}
v(f, g)(x) = \int_{x_0}^x \exp \Big( - \int_{x_0}^y 2 f(z) \dd z \Big) \int_{x_0}^y \frac{2 \exp (\int_{x_0}^u2 f(z) \dd z)}{g(u)} \dd u \dd y,\quad x \in I,
\end{align}
for a fixed \(x_0 \in I\). 
The main result of this section is the following:
\begin{theorem}\label{theo: 1D Feller}
		Suppose that 
		\begin{align} \label{eq: U1 ass}
		\lim_{x \nearrow r} v \left( \overline{u}, \overline{a}\right)(x) = 
		\lim_{x \searrow l} v\left(\underline{u}, \overline{a}\right)(x) =\infty.
		\end{align}
		Moreover, for all \(n \in \mathbb{N}\) assume that for \(\llambda \otimes \Q^n\)-a.a. \((t, \omega) \in [0, T] \times \Omega\)
		\begin{equation}\label{eq: to hold}
		\begin{split}
		\sigma^2_t (\omega) &\leq \zeta (t) \overline{a} (S_t (\omega)), \\
		b_t (\omega) &\leq \sigma_t^2 (\omega) \overline{u} (S_t(\omega)),\\
		b_t (\omega) &\geq \sigma^2_t (\omega) \underline{u} (S_t(\omega)).
		\end{split}
		\end{equation}
		Then, \eqref{eq: cond qn} holds. 
	\end{theorem}
\begin{proof}
Due to \cite[Lemma 5.5.26]{KaraShre}, there are differentiable functions \(U_1 \colon [x_0, r) \to [1, \infty)\) and \(U_2 \colon (l, x_0] \to [1, \infty)\) with locally absolutely continuous derivatives and a \(\llambda\)-null set \(N' \subset I\) such that \(U_1\) is increasing, \(U_2\) is decreasing, \(U_1 (x_0) = U_2(x_0) = 1, U'_1 (x_0) = U'_2(x_0) = 0\), for all \(x \in [x_0, r) \backslash N'\) and for all \(y \in (l, x_0] \backslash N'\)
\begin{align*}
\overline{a} (x) \left(\tfrac{1}{2}  U_1'' (x) + \overline{u} U_1' (x)\right) &= U_1 (x)\quad \textup{ and }\quad
 \overline{a} (y) \left(\tfrac{1}{2} U_2'' (y) + \underline{u} U_2'  (y)\right)= U_2 (y),
\end{align*}
\(1 + v(\overline{u}, \overline{a}) \leq U_1\) 
and \(1 + v(\underline{u}, \overline{a}) \leq U_2\). 
We define 
\begin{align*}
V \triangleq \begin{cases}
U_1,&\textup{ on } [x_0, r),
\\
U_2,&\textup{ on } (l,x_0],
\end{cases}
\end{align*}
which is a differentiable function with locally absolutely continuous derivative. In particular, \(V' \geq 0\) on \([x_0, r)\), \(V' \leq 0\) on \((l, x_0]\),
\(
\frac{1}{2} V'' + \underline{u} V' \geq 0
\)
on \((l, x_0] \backslash N'\) and \(
\frac{1}{2} V'' + \overline{u} V' \geq 0
\) on \([x_0, r)\backslash N'\).
Furthermore, \begin{align*}\lim_{x \nearrow r} V(x) = \lim_{x\searrow l} V(x) = \infty,\end{align*} due to the assumption \eqref{eq: U1 ass}.
 Let \(\widetilde{N}\) be the set of all \((t, \omega) \in [0, T] \times \Omega\) such that \eqref{eq: to hold} holds. 
For all \((t, \omega) \in \widetilde{N} \colon S_t (\omega) \in [x_0, r) \backslash N'\) 
	\begin{align*}
		\mathcal{L}V (t) (\omega) 
		&= \tfrac{1}{2}  \sigma^2_t(\omega) V'' (S_t(\omega)) + b_t(\omega) V' (S_t(\omega))
		\\&\leq \sigma^2_t(\omega) \left( \tfrac{1}{2} V'' (S_t(\omega)) + \overline{u} (S_t(\omega)) V' (S_t(\omega))\right)
		\\&\leq \zeta (t) \overline{a}(S_t(\omega)) \left(\tfrac{1}{2} V'' (S_t(\omega))+ \overline{u} (S_t(\omega))V' (S_t(\omega))\right)
		= \zeta (t) V (S_t(\omega)).
	\end{align*}
	In the same manner we see that for all \((t, \omega) \in \widetilde{N}\colon S_t (\omega) \in (l, x_0] \backslash N'\) 
	\[
	\mathcal{L} V(t) (\omega) \leq \zeta(t) V(S_t(\omega)).
	\]
	We conclude that \eqref{ineq: Lyapunov2} holds for \(N = N'\). 
The claim follows from Theorem \ref{theo: NE 1}.
\end{proof}
\subsection{Criteria for non-existence}
In this section we give a converse to Theorem  \ref{theo: 1D Feller}. As in Section \ref{sec: MP SE}, let \(I = (l, r)\) with \(- \infty \leq l < r \leq + \infty\)  and let \(\underline{a} \colon I \to (0, \infty)\) and \(\underline{u}, \overline{u} \colon I \to \mathbb{R}\) be Borel functions such that 
\[
\frac{1}{\underline{a}} + |\underline{u}| + |\overline{u}| \in L^1_\textup{loc}(I).
\]
If \((f, g)\) is one of the pairs \((\underline{u}, \underline{a}), (\overline{u}, \underline{a})\), we set \(v(f, g)\) as in \eqref{eq: v}.

Let \(0 < T < \infty\), \((\Omega, \mathcal{F})\) be a measurable space with right-continuous filtration \(\F = (\mathcal{F}_t)_{t \in [0, T]}\) and \(s_0 \in I\). Suppose that \((\Omega, \mathcal{F}, \F)\) supports three progressively measurable processes \(S = (S_t)_{t \in [0, T]}, b = (b_t)_{t \in [0, T]}\) and \(\sigma = (\sigma_t)_{t \in [0, T]}\). We define \(\mathcal{I}\) be the set of all pairs \((\Q, B)\) consisting of a probability measure \(\Q\) on \((\Omega, \mathcal{F})\) and an \((\F, \Q)\)-Brownian motion \(B = (B_t)_{t \in [0, T]}\) with the properties that \(S\) is \(\Q\)-a.s. \(I\)-valued and
\[
\dd S_t = b_t \dd t + \sigma_t \dd B_t,\quad S_0 = s_0,
\]
where it is implicit that the integrals are well-defined.

\begin{theorem} \phantomsection \label{theo: 1D Feller p2} 
	\begin{enumerate}		
		\item[\textup{(i)}] 
		Suppose that the pair \((\underline{u}, \underline{a})\) satisfies the YW conditions (see Section \ref{sec: GC} for this terminology) and 
		\begin{align*}
		\lim_{x \nearrow r} v \left( \underline{u}, \underline{a}\right)(x) < \infty.  \end{align*}
		Then, there exists no pair \((\Q, B) \in \mathcal{I}\) such that for \(\llambda \otimes \Q\)-a.a. \((t, \omega) \in [0, T] \times \Omega\)
		\begin{equation}\label{eq: contr 1}
		\begin{split}
		\underline{a} (S_t(\omega)) &\leq \sigma^2_t(\omega),\\
		\underline{u} (S_t(\omega)) \sigma^2_t(\omega) &\leq b_t(\omega).
		\end{split}
		\end{equation}
		\item[\textup{(ii)}]
		Suppose that the pair \((\overline{u}, \underline{a})\) satisfies the YW conditions and 
		\[
		\lim_{x \searrow l} v\left(\overline{u}, \underline{a}\right)(x) < \infty.
		\]
		Then, there exists no pair \((\Q, B)\in \mathcal{I}\) such that for \(\llambda \otimes \Q\)-a.a. \((t, \omega) \in [0, T] \times \Omega\)
		\begin{equation}\label{eq: contr 2}
		\begin{split}
		\underline{a} (S_t(\omega)) &\leq \sigma^2_t(\omega),\\
		\overline{u} (S_t(\omega)) \sigma^2_t(\omega) &\geq b_t(\omega).
		\end{split}
		\end{equation}
	\end{enumerate}
\end{theorem}
\begin{proof}
\textbf{(i).} 
We use a comparison and contradiction argument as in the proof of \cite[Theorem 4.1]{criens17b}. 
For contradiction, assume that \((\Q, B)\in \mathcal{I}\) is such that \eqref{eq: contr 1} holds. W.l.o.g. we assume that \(\F\) is \(\Q\)-complete. In the following we work on \((\Omega, \mathcal{F}, \F, \Q)\).
Because \(\underline{a}\) is positive and continuous and a.s \[\llambda (\{t \in [0, T] \colon \underline{a}(S_t) > \sigma^2_t\}) = 0, \qquad \int_0^T \sigma^2_s \dd s < \infty,\] the function \[[0, T] \ni t \mapsto \int_0^t \frac{\sigma^2_s}{\underline{a}(S_s)}\dd s\] is a.s. finite, continuous and strictly increasing, which implies that the same holds for the function
\[
\phi_t \triangleq \inf \Big( s \in [0, T] \colon \int_0^s \frac{\sigma^2_r}{\underline{a}(S_r)} \dd r \geq t \Big),\quad t \in [0, T],
\]
see \cite[pp. 179 -- 180]{RY}.
Furthermore, we have a.s. \(\phi_t \leq t \text{ for all } t \in [0, T]\). We redefine \(\phi_t\) to be zero on the null sets where the previously mentioned properties fail. 
Because \(\F\) is complete, this modification of \((\phi_t)_{t \in [0, T]}\) is an increasing and continuous sequence of finite stopping times. 

Next, we set \(\F_\phi \triangleq (\mathcal{F}_{\phi_t})_{t \in [0, T]}\).
The following lemma follows from \cite[Propositions V.1.4, V.1.5]{RY}. 
\begin{lemma}\label{lem: tc}
	Suppose that \((H_t)_{t \in [0, T]}\) is  progressively measurable. Then, the time-changed process \((H_{\phi_t})_{t \in [0, T]}\) is \(\F_\phi\)-progressively measurable and a.s.
	\[
	\int_0^t H_{\phi_s} \dd s = \int_0^{\phi_t} \frac{H_s \sigma^2_s}{\underline{a}(S_s)} \dd s, \quad t \in [0, T],
	\]
	provided the integrals are well-defined.
	Moreover, the process \(B_\phi  =  (B_{\phi_t})_{t \in [0, T]}\) is a continuous local \(\F_\phi\)-martingale with a.s.
	\(
	[B_\phi, B_\phi]= \phi
	\),
	and if a.s. \(\int_0^T H^2_s \dd s < \infty\) then also a.s. \(\int_0^T H_{\phi_s}^2 \dd \phi_s < \infty\) and a.s.
	\[
	\int_0^t H_{\phi_s} \dd B_{\phi_s} = \int_0^{\phi_t} H_s \dd B_s, \quad t \in [0, T].
	\]
\end{lemma}
We deduce from Lemma \ref{lem: tc} 
that a.s.
\begin{align*}
\llambda \big(\big\{ t \in [0, T] \colon \underline{a}(S_{\phi_t}) &> \sigma^2_{\phi_t} \text{ or } \underline{u}(S_{\phi_t}) \sigma^2_{\phi_t} > b_{\phi_t}\big\}\big) \\&= \int_0^{\phi_T} \frac{\1_{\{\underline{a}(S_s) > \sigma^2_s\} \cup \{\underline{u}(S_s) \sigma^2_s > b_s\}} \sigma^2_s}{\underline{a}(S_s)} \dd s = 0. 
\end{align*}
We will use this observation in the following without further reference. 

Applying Lemma \ref{lem: tc} with
\[
H_t \triangleq \frac{\underline{a}(S_t)}{\sigma^2_t} \1_{\{\sigma^2_t > 0\}}, \quad t \in [0, T], 
\]
yields that a.s.
\begin{align}\label{eq: meas eq}
\dd \phi_t = \frac{\underline{a}(S_{\phi_t})}{\sigma^2_{\phi_t}} \dd t.
\end{align}

Using again Lemma \ref{lem: tc}, we obtain that a.s. for all \(t \in [0, T]\)
\begin{align*}
S_{\phi_t} &= S_{\phi_0} + \int_0^{\phi_t} b_{s}\dd s + \int_0^{\phi_t} \sigma_s \dd B_s
\\&= s_0 + \int_0^t \frac{b_{\phi_s}\underline{a}(S_{\phi_s})}{\sigma^2_{\phi_s}} \dd s + \int_0^t \sigma_{\phi_s} \dd B_{\phi_s}
\\&= s_0 + \int_0^t  \frac{b_{\phi_s}\underline{a}(S_{\phi_s})}{\sigma^2_{\phi_s}} \dd s + \int_0^t \underline{a}^\frac{1}{2} (S_{\phi_s}) \dd B'_s,
\end{align*}
where 
\[
B' \triangleq \int_0^\cdot \frac{\sigma_{\phi_s} \dd B_{\phi_s}}{\underline{a}^\frac{1}{2} (S_{\phi_s})}.
\]
Due to Lemma \ref{lem: tc} and \eqref{eq: meas eq}, we obtain that a.s. for all \(t \in [0, T]\)
\begin{align*}
[B', B']_t &= \int_0^t \frac{\sigma^2_{\phi_s}}{\underline{a}(S_{\phi_s})} \dd [B_{\phi}, B_{\phi}]_s \\&= \int_0^t \frac{\sigma^2_{\phi_s}}{\underline{a}(S_{\phi_s})} \dd \phi_s
\\ &= \int_0^t \frac{\sigma^2_{\phi_s}}{\underline{a}(S_{\phi_s})} \frac{\underline{a}(S_{\phi_s})}{\sigma^2_{\phi_s}} \dd s
 = t.
\end{align*}
Consequently, \(B'\) is a continuous local \(\F_\phi\)-martingale with a.s. \([B', B']_t = t\) for \(t \in [0, T]\), i.e. an \(\F_\phi\)-Brownian motion due to L\'evy's characterization.
We summarize that
\begin{align*}
\dd S_{\phi_t} =\underline{a} (S_{\phi_t}) \frac{b_{\phi_t}}{\sigma^2_{\phi_t}} \dd t + \underline{a}^\frac{1}{2}(S_{\phi_t}) \dd B'_t, \quad S_{\phi_0} = s_0.
\end{align*}
Using a standard extension of \((\Omega, \mathcal{F}, \F_{\phi}, \Q)\) we can extend \((B'_t)_{t \in [0, T]}\) to a Brownian motion \((B'_t)_{t \geq 0}\), see, e.g., the proof of \cite[Theorem V.1.7]{RY}.

We will use the following terminology: When we say that \((V_t)_{t \geq 0}\) is a continuous \([l, r]\)-valued process we mean that all its paths are continuous in the \([l, r]\)-topology \emph{and} absorbed in \(\{l, r\}\), i.e. that \(V_t = V_{\tau(V)}\) for all \(t \geq \tau(V) \triangleq \inf(t \in \mathbb{R}_+ \colon V_t \not \in I)\). This convention is in line with \cite[Definition 5.5.20]{KaraShre}.
\begin{definition}\label{def: see}
	Let \(\mu \colon I \to \mathbb{R}\) and \(v \colon I \to \mathbb{R}\) be Borel functions.
	We say that an SDE
	\begin{align}\label{eq: def SDE}
	\dd V_t = \mu(V_t) \dd t + v(Y_t)\dd B^*_t,
	\end{align}
	where \((B^*_t)_{t \geq 0}\) is a one-dimensional Brownian motion,
	satisfies \emph{strong existence and uniqueness up to explosion}, if on any complete probability space \((\Omega^o, \mathcal{F}^o, \p^o)\) with complete right-continuous filtration \(\F^o = (\mathcal{F}^o_t)_{t \geq 0}\), which supports a Brownian motion \((B^*_t)_{t \geq 0}\) and an \(I\)-valued \(\mathcal{F}^o_0\)-measurable random variable \(\psi\), there exists a up to indistinguishability unique adapted continuous \([l, r]\)-valued process \((V_t)_{t \geq 0}\) such that a.s.
	\[
	V_{t \wedge \theta_n} = \psi + \int_0^{t \wedge \theta_n} \mu(V_s)\dd s + \int_0^{t \wedge \theta_n} v(V_s)\dd B^*_s, \quad t \geq 0, n \in \mathbb{N},
	\]
	where 
	\[
	\theta_n \triangleq \inf(t \in \mathbb{R}_+ \colon V_t \not \in (l_n, r_n)), \quad n \in \mathbb{N}.
	\]
	It is implicit that the integrals are well-defined.
	The process \((V_t)_{t \geq 0}\) is called the \emph{solution process to \eqref{eq: def SDE} with driver \((B^*_t)_{t \geq 0}\)}.
\end{definition} 

Due to  \cite[Remark 4.50 (2), Theorem 4.53]{MANA:MANA19911510111},
the SDE 
\begin{align}\label{eq: some SDE}
\dd V_t = \underline{a}(V_t) \underline{u}(V_t)\dd t + \underline{a}^\frac{1}{2} (V_t)\dd B^*_t
\end{align}
satisfies strong existence and uniqueness up to explosion.

Consequently, there exists a solution process \((Y_t)_{t \geq 0}\) to \eqref{eq: some SDE} with driver \((B'_t)_{t \geq 0}\).
The following lemma is proven after the proof of Theorem \ref{theo: 1D Feller p2} is complete.
\begin{lemma}\label{lem: order}
	Almost surely \(Y_t \leq S_{\phi_t}\) for all \(t \leq T \wedge \tau(Y)\).
\end{lemma}
Because \((Y_t)_{t \geq 0}\) is regular due to \cite[Proposition 2.2]{mijatovic2012b} and \(\lim_{x \nearrow r} v \left( \underline{u}, \underline{a}\right)(x) < \infty\), we deduce from \cite[Proposition 2.12]{mijatovic2012b} and \cite[Theorem 1.1]{bruggeman2016} that \((Y_t)_{t \in [0, T]}\) reaches \(r\) with positive probability. Consequently, due to Lemma \ref{lem: order}, \((S_t)_{t \in [0, T]}\) reaches \(r\) with positive probability. This is a contradiction. 

\textbf{(ii).} 
For contradiction, assume that \((\Q, B)\in \mathcal{I}\) is such that \eqref{eq: contr 2} holds.
By the same arguments as in part (i), there exists a process \((Y_t)_{t \geq 0}\) such that
\[
\dd Y_t =  \underline{a}(Y_t) \overline{u} (Y_t) \dd t +\underline{a}^\frac{1}{2}(Y_t) \dd B'_t, \quad Y_0 = s_0, 
\]
and a.s. \(S_{\phi_t} \leq Y_t\) for all \(t \leq T \wedge \tau(Y).\) Because \(\lim_{x \searrow l} v\left(\overline{u}, \underline{a}\right)(x) <\infty,\) the process \((Y_t)_{t \in [0, T]}\) reaches \(l\) with positive probability and again the pathwise ordering gives a contradiction. 
\end{proof}
\noindent
\textit{Proof of Lemma \ref{lem: order}:}
There are functions \(h_n \in \mathscr{H}\) and \(\kappa_n \in \mathscr{K}\) such that for all \(x, y \in [l_n, r_n]\) 
\begin{align*}
|\underline{a}^\frac{1}{2} (x) - \underline{a}^\frac{1}{2} (y)| &\leq h_n (|x - y|),\quad |\underline{a}(x) \underline{u}(x) - \underline{a}(y) \underline{u}(y)| \leq \kappa_n (|x - y|).
\end{align*}
We set \[\rho_n \triangleq \inf (t \in [0, T]  \colon S_{\phi_t} \not \in (l_n, r_n) \textup{ or } Y_t \not \in (l_n, r_n)).\]
Note that for all \(t \in (0, T]\) we have 
\[
\int_0^{t \wedge \rho_n} \frac{\dd [Y - S_{\phi},Y - S_{\phi}]_s}{h^2_n(|Y_s -S_{\phi_s}|)} = \int_0^{t \wedge \rho_n} \frac{\big(\underline{a}^\frac{1}{2} (Y_s) - \underline{a}^\frac{1}{2} (S_{\phi_s})\big)^2}{h^2_n(|Y_s - S_{\phi_s}|)} \dd s \leq \int_0^t \dd s = t.
\]
Thus, \cite[Lemma IX.3.3]{RY} implies that the local time of \(Y_{\cdot \wedge \rho_n}- S_{\phi_{\cdot \wedge \rho_n}}\) in the origin is a.s. zero.
We deduce from Tanaka's formula that a.s.
\[
(Y_{t \wedge \rho_n} - S_{\phi_{t \wedge \rho_n}})^+ = \int_0^{t \wedge \rho_n} \1_{\{Y_s - S_{\phi_s} > 0\}} \dd (Y_s - S_{\phi_s}), \quad t \in [0, T].
\]
Taking expectation, using the martingale property of the Brownian part of \(Y_{\cdot \wedge \rho_n} - S_{\phi_{\cdot \wedge \rho_n}}\) and Jensen's inequality yields that for all \(t \in [0, T]\)
\begin{align*}
\E^\Q \big[ (Y_{t \wedge \rho_n} - S_{\phi_{t \wedge \rho_n}})^+\big] &= \E^\Q\Big[ \int_0^{t \wedge \rho_n} \1_{\{Y_s - S_{\phi_s} > 0\}} \Big( \underline{a} (Y_s) \underline{u}(Y_s) - \underline{a}(S_{\phi_s}) \frac{b_{\phi_s}}{\sigma^2_{\phi_s}}\Big)\dd s \Big]
\\&\leq \E^\Q\Big[ \int_0^{t \wedge \rho_n} \1_{\{Y_s - S_{\phi_s} > 0\}} \big| \underline{a} (Y_s) \underline{u}(Y_s) - \underline{a}(S_{\phi_s}) \underline{u}(S_{\phi_s})\big|\dd s \Big]
\\&\leq \E^\Q\Big[ \int_0^{t \wedge \rho_n} \1_{\{Y_s - S_{\phi_s} > 0\}} \kappa_n (|Y_s - S_{\phi_s}|)\dd s \Big]
\\&\leq \int_0^t \E^\Q\big[ \kappa_n ((Y_{s \wedge \rho_n} - S_{\phi_{s \wedge \rho_n}})^+)\big] \dd s
\\&\leq \int_0^t \kappa_n \big( \E^\Q\big[(Y_{s \wedge \rho_n} - S_{\phi_{s \wedge \rho_n}})^+\big] \big)\dd s.
\end{align*}
Finally, Bihari's lemma (see \cite[Lemma E.2]{criens17b}) yields that for all \(t \in [0, T]\)
\[
\E^\Q \big[ (Y_{t \wedge \rho_n} - S_{\phi_{t \wedge \rho_n}})^+\big] = 0.
\]
Consequently, due to the continuous paths of \(Y\) and \(S_\phi\), the claim follows.
\qed

\subsection{Proof of Theorem \ref{theo: mart Ito}}
Because non-negative local martingales are supermartingales, \(Z\) is a martingale if and only if  \(\E^\p[Z_T] = 1\).
By (M1), we can define \(\Q^n\) by the Radon--Nikodym derivative \(\frac{\dd \Q^n}{\dd \p} = Z_{T \wedge \tau_n}\). 
We note that the assumption \(\llambda \otimes \p\)-a.e. \(\sigma \not = 0\) implies \eqref{eq: nondeg n}.
Due to Girsanov's theorem, there exists a \(\Q^n\)-Brownian motion \(B^n = (B^n_t)_{t \in [0, T]}\) such that 
\[
\dd S_{t \wedge \tau_n} = (b_t + c_t \sigma_t) \1_{\{t \leq \tau_n\}} \dd t + \sigma_t \1_{\{t \leq \tau_n\}} \dd B^n_t.
\]
The monotone convergence theorem yields that 
\begin{align*}
\E^\p \big[ Z_T\big] &= \limsup_{n \to \infty} \E^\p \big[ Z_T \1_{\{\tau_n = \infty\}}\big]
\\&= \limsup_{n \to \infty} \Q^n(\tau_n = \infty).
\end{align*}
In view of (M2) and (M3), Theorem \ref{theo: 1D Feller} yields that 
\[
\limsup_{n \to \infty} \Q^n (\tau_n = \infty) = 1.
\]
Thus, \(\E^\p[Z_T] = 1\) and the proof is complete.
\qed

\subsection{Proof of Theorem \ref{theo: general SLM}}
For contradiction, assume that \((Z_t)_{t \in [0, T]}\) is a martingale. 
Define  a probability measure \(\Q\) by the Radon--Nikodym derivative \(\frac{\dd \Q}{\dd \p} \triangleq Z_T\). By Girsanov's theorem, there exists a \(\Q\)-Brownian motion \(B = (B_t)_{t \in [0, T]}\) such that 
\[
\dd S_t =  (b_t + c_t \sigma_t)\dd t + \sigma_t \dd B_t.
\]
Consequently, in case (SL1) holds we obtain a contradiction to part (i) of Theorem \ref{theo: 1D Feller p2} and in case (SL2) holds we obtain a contradiction to part (ii) of Theorem \ref{theo: 1D Feller p2}. The proof is complete.
\qed

\section{Proof of Theorem \ref{theo: mart MS}}\label{sec: pf MS}
The section is split into two parts: First, we prove existence, non-existence and local uniqueness for switching diffusions and second, we deduce Theorem \ref{theo: mart MS}.

\subsection{Existence and non-existence criteria}\label{sec: jump type}
As in Section \ref{sec: MG MS}, let \(I = (l, r)\) with \(- \infty \leq l < r \leq + \infty\) and  \(J = \{1, \dots, N\}\) with \(1 \leq N \leq \infty\). Moreover, let \(u \colon I \times J \to \mathbb{R}\) and \(\sigma \colon I \times J \to \mathbb{R}\backslash \{0\}\) be Borel functions such that 
\begin{align}\label{eq: ES cond}
\frac{1 + u(\cdot, j)}{\sigma^2 (\cdot, j)} \in L^1_\textup{loc} (I) \text{ for all } j \in J.
\end{align}
We fix \(x_0 \in I\) and set 
\[
v(x, j) \triangleq \int_{x_0}^x \exp \left( - \int_{x_0}^y  \frac{2u(z, j)}{\sigma^2(z, j)}  \dd z \right) \int_{x_0}^y \frac{2 \exp (\int_{x_0}^s \frac{2u(z, j)}{\sigma^2(z, j)} \dd z)}{\sigma^2(s, j)} \dd s \dd y
\]
for \((x, j) \in I \times J\).
Let \((\Omega, \mathcal{F}, \F, \p)\) be a filtered complete probability space with a right-continuous and complete filtration \(\F = (\mathcal{F}_t)_{t \geq 0}\), which supports a Brownian motion \(W = (W_t)_{t \geq 0}\), a \(J\)-valued irreducible continuous-time Feller--Markov chain \(\xi = (\xi_t)_{t \geq 0}\) and an \(I\)-valued \(\mathcal{F}_0\)-measurable random variable \(\phi\).
The main result of this section is the following:
\begin{theorem}\phantomsection\label{theo: existence Markov}
	\begin{enumerate} 
		\item[\textup{(i)}] Suppose that \(\sigma\) satisfies the ES conditions for all \(j \in J\) (see Section \ref{sec: MG MS} for this terminology) and
		that \begin{align}\label{eq: MC FT} \lim_{x \searrow l} v(x, j) = \lim_{x \nearrow r} v(x, j) = \infty \text{ for all } j \in J.\end{align} Then, there exists an adapted \(I\)-valued continuous process \((Y_t)_{t \geq 0}\) such that 
		\begin{align}\label{eq: SDE MDP}
		Y = \phi + \int_0^\cdot u(Y_s, \xi_s)\dd s + \int_0^\cdot \sigma (Y_s, \xi_s) \dd W_s, 
		\end{align}
		where it is implicit that the integrals are well-defined.
		\item[\textup{(ii)}] 
		Assume there exists a \(j \in J\) such that \(\sigma\) satisfies the ES conditions for \(j\) 
		and 
		\[\lim_{x \searrow l} v(x, j) < \infty \text{ or } \lim_{x \nearrow r} v(x, j) <  \infty.\] 
		Let \(0 < T \leq \infty\) be a time horizon. If \(\xi\) is recurrent, then there exists no adapted \(I\)-valued continuous process \(Y = (Y_t)_{t \in [0, T]}\) such that \eqref{eq: SDE MDP} holds.
	\end{enumerate}
\end{theorem}
\begin{proof}
	The case \(N = 1\) concerns classical diffusions for which all claims are known, see \cite{bruggeman2016,MANA:MANA19911510111, KaraShre} for details. We prove the claim under the assumption \(N \geq 2\). 
	 
		\textbf{(i).}
	We define the jump times of \(\xi\) inductively by
	\[
	\gamma_0 \triangleq 0, \quad \gamma_n \triangleq \inf(t \geq \gamma_{n-1} \colon \xi_t \not = \xi_{\gamma_{n-1}}), \quad n \in \mathbb{N}.
	\]
	Because \(\xi\) is irreducible, we have a.s. \(\gamma_n < \infty\) (see \cite[Theorem 10.19]{Kallenberg}) and a.s. \(\gamma_n - \gamma_{n-1} > 0\) for all \(n \in \mathbb{N}\).
	
	We follow the idea from the proof of \cite[Theorem IV.9.1]{IW} and construct the process \(Y\) explicitly from solutions to the SDEs
	\begin{align}\label{eq: SDE 1}
	\dd X^j_t = u(X^j_t, j) \dd t + \sigma (X^j_t, j)\dd W'_t,
	\end{align}
	where \(W' = (W'_t)_{t \geq 0}\) is a Brownian motion.
	For the construction we require a strong existence and uniqueness property, which we explain next.
	
	Fix \(j \in J\). It follows from \cite[Remark 4.50 (2), Theorem 4.53]{MANA:MANA19911510111} and Feller's test for explosion (see \cite[Theorem 5.5.29]{KaraShre}) that the SDE \eqref{eq: SDE 1} has a weak solution and that it satisfies pathwise uniqueness for all deterministic initial values. We conclude from \cite[Theorem 18.14]{Kallenberg} that there exists a Borel function \(F^{j} \colon I \times C(\mathbb{R}_+, \mathbb{R}) \to C(\mathbb{R}_+, I)\) such that for any one-dimensional Brownian motion \(W' = (W'_t)_{t \geq 0}\) and any \(I\)-valued random variable \(\psi\), which is independent of \(\sigma(W'_t, t \in \mathbb{R}_+)\), the process \(X^j = F^{j}(\psi, W')\) is a solution process to \eqref{eq: SDE 1} with \(X^j_0 = \psi\), which is adapted to the completion of the natural filtration of \(W'\) and \(\psi\), see \cite[Definition 5.2.1]{KaraShre}. The function \(F^j\) is independent of the law of \(\psi\) and universally adapted (see \cite[p. 346]{Kallenberg} for a definition). 
	
	Set \(W^n \triangleq W_{\cdot + \gamma_n} - W_{\gamma_n}\). Due to \cite[Proposition V.1.5]{RY} and L\'evy's characterization, \(W^n\) is a Brownian motion for the filtration \(\F^n \triangleq (\mathcal{F}_{t + \gamma_n})_{t \geq 0}\).
	In particular, \(W^n\) is independent of \(\mathcal{F}_{\gamma_n}\).
	By induction, define \begin{align*}
	Y^0 &\triangleq \sum_{j \in J} F^{j}(\phi, W) \1_{\{\xi_0 = j\}}, \\ Y^n &\triangleq \sum_{j \in J} F^{j}(Y^{n-1}_{\gamma_n - \gamma_{n-1}}, W^n) \1_{\{\xi_{\gamma_n} = j\}}, \quad n \in \mathbb{N}.
	\end{align*}
	Moreover, set 
	\[
	Y_t \triangleq \sum_{n = 0}^\infty Y^n_{t - \gamma_n} \1_{\{\gamma_n \leq t < \gamma_{n+1}\}},\quad t \in \mathbb{R}_+.
	\]
	The process \(Y\) is \(I\)-valued and continuous and, as we explain next, it is also \(\F\)-adapted.
	Define \(H_t \triangleq Y^n_{t - \gamma_n} \1_{\{\gamma_n < t\}}\). 
	We claim that \((H_t)_{t \geq 0}\) is \(\F\)-progressively measurable. Because \(t \mapsto Y^n_{t - \gamma_n} \1_{\{\gamma_n < t\}}\) is left-continuous and \(s \mapsto  Y^n_{t - s} \1_{\{s < t\}}\) is right-continuous, an approximation argument shows that is suffices to explain that \((h_t)_{t \geq 0} \triangleq (Y^n_{t - \zeta} \1_{\{\zeta < t\}})_{t \geq 0}\) is \(\F\)-adapted for any \(\F\)-stopping time \(\zeta\) which takes values in the countable set \(2^{-m} \overline{\mathbb{N}}\) for some \(m \in \mathbb{N}\) and satisfies \(\zeta \geq \gamma_n\).
	Let \(G \in \mathcal{B}(\mathbb{R})\) and set \(N_{m, t} \triangleq 2^{-m} \overline{\mathbb{N}}\cap [0, t)\). We have 
	\[
	\{h^m_t \in G\} = \Big(\bigcup_{k \in N_{m, t}}  \big( \{h^m_t \in G\} \cap \{\zeta = k\}\big)\Big) \cup \big(\{0 \in G\} \cap \{\zeta \geq t\}\big) \in \mathcal{F}_t.
	\]
	Here, we use that \(\{Y^n_{t - k} \in G\} \in \mathcal{F}_{t - k + \gamma_n} \subseteq \mathcal{F}_{t - k + \zeta}\) and that \(\mathcal{F}_{t - k + \zeta} \cap \{\zeta = k\} \in \mathcal{F}_t\). Thus, \((H_t)_{t \geq 0}\) is \(\F\)-progressively measurable and consequently \((Y_t)_{t \geq 0}\) is \(\F\)-adapted.

	We note that 
	\[
	\gamma_n - \gamma_{n - 1} = \inf(t \in \mathbb{R}_+ \colon \xi_{t + \gamma_{n - 1}} \not = \xi_{\gamma_{n-1}}),
	\]
	which is an \(\F^{n - 1}\)-stopping time. Thus, \(Y^{n-1}_{\gamma_n - \gamma_{n - 1}}\) is \(\mathcal{F}_{\gamma_n}\)-measurable and therefore independent of \(\sigma(W^n_t, t \in \mathbb{R}_+)\). This yields that the process \(X^{n, j} \triangleq F^{j} (Y^{n - 1}_{\gamma_n - \gamma_{n-1}}, W^n)\) satisfies
	\begin{align*}
	\dd X^{n, j}_t &= u(X^{n, j}_t, j) \dd t + \sigma (X^{n, j}_t, j) \dd W^n_t, \quad X^{n, j}_0 = Y^{n-1}_{\gamma_n - \gamma_{n -1}}.
	\end{align*}
	Thus, due to classical rules for time-changed stochastic integrals (see \cite[Propositions V.1.4, V.1.5]{RY}), a.s. for \(t \in [\gamma_n, \gamma_{n+1}]\) on \(\{\xi_{\gamma_n} = j\}\) we have
	\begin{align}
	Y^n_{t - \gamma_n} 
	&= Y^{n-1}_{\gamma_n - \gamma_{n - 1}} + \int_0^{t - \gamma_n} u(X^{n, j}_s, j)\dd s+ \int_0^{t - \gamma_n} \sigma(X^{n, j}_s, j) \dd W^n_s \nonumber
	\\&= Y^{n-1}_{\gamma_n - \gamma_{n-1}} + \int_{\gamma_n}^{t} u(Y^n_{s - \gamma_n}, j)\dd s + \int_{\gamma_n}^t \sigma (Y^n_{s - \gamma_n}, j) \dd W_s \nonumber
	\\&= Y^{n-1}_{\gamma_n - \gamma_{n-1}} + \int_{\gamma_n}^{t} u(Y_s, \xi_{s})\dd s + \int_{\gamma_n}^t \sigma (Y_s, \xi_{s}) \dd W_s. \nonumber
	\end{align}
	By induction, a.s. for \(t \in [\gamma_n, \gamma_{n + 1}]\)
	\[
	Y^n_{t - \gamma_n} = \phi + \int_0^t u(Y_s, \xi_{s}) \dd s + \int_0^t \sigma (Y_s, \xi_{s})\dd W_s.
	\]
	Therefore, the process \(Y\) satisfies the SDE
	\[
	\dd Y_t = u(Y_t, \xi_{t}) \dd t + \sigma(Y_t, \xi_{t}) \dd W_t,\quad S_0 = \phi,
	\]
	and the proof of (i) is complete.

	\textbf{(ii).}
	For contradiction, assume that \(Y\) satisfies \eqref{eq: SDE MDP}.
	Let \(j \in J\) be such that \(\lim_{x \searrow l} v(x, j) < \infty\) or \(\lim_{x \nearrow r} v(x, j) < \infty\).
	We define 
	\[
	\delta \triangleq \inf (t \in \mathbb{R}_+ \colon \xi_t = j ), \quad \zeta \triangleq \inf(t \geq \delta \colon \xi_t \not = j).
	\]
	Because \(\xi\) is recurrent, we have a.s. \(\delta < \infty\), see \cite[Theorem 1.5.7]{norris_1997}.
	Due to the strong Markov property of \(\xi\) and \cite[Lemma 10.18]{Kallenberg}, for all \(G \in \mathcal{B}(\mathbb{R}_+)\) it holds that
	\begin{align}\label{eq: exp}
	\p (\zeta - \delta \in G) = - \int_G q_{jj} e^{q_{jj} x} \dd x,
	\end{align}
	where \(q_{jj} < 0\) is the \(j\)-th diagonal element of the \(Q\)-matrix of \(\xi\). 
	
	Recall our convention that we call a process \(V = (V_t)_{t \geq 0}\) to be continuous and \([l, r]\)-valued in case all paths are continuous in the \([l, r]\)-topology \emph{and} absorbed in \(\{l, r\}\), i.e. that \(V_t = V_{\tau(V)}\) for all \(t \geq \tau(V) \triangleq \inf(t \in \mathbb{R}_+ \colon V_t \not \in I)\).
	
	It follows from \cite[Remark 4.50 (2), Theorem 4.53]{MANA:MANA19911510111} that the SDE \eqref{eq: SDE 1}
	satisfies strong existence and uniqueness up to explosion in the sense of Definition \ref{def: see}.
	
	Consequently, there exists a continuous \([l, r]\)-valued process \(X = (X_t)_{t \geq 0}\) such that
	\begin{align}\label{eq: SDE e}
	\dd X_t = u(X_t, j) \dd t + \sigma (X_t, j) \dd W^\delta_t,\quad X_0 = Y_{\delta \wedge T},
	\end{align}
	where \(W^\delta \triangleq W_{\cdot + \delta \wedge T} - W_{\delta \wedge T}\) is a Brownian motion for the filtration \(\F^\delta \triangleq (\mathcal{F}_{t + \delta \wedge T})_{t \geq 0}\).
	We prove the following lemma after the proof of (ii) is complete.
	\begin{lemma}\label{lem: comparison}
		Almost surely \(Y_{t + \delta} = X_t\) for all \(0 \leq t \leq \zeta - \delta\) on \(\{\zeta \leq T\}\).
	\end{lemma}
	Because on \(\{\tau (X) < \infty\}\) we have \(X_{\tau(X)} \not \in I\), Lemma \ref{lem: comparison} implies that
	\begin{equation}\label{eq: long comp} \begin{split}
	\p (\tau (X) \leq \zeta - \delta, \zeta \leq T) = 0.
	\end{split}\end{equation}
	The proof of the following lemma is given after the proof of (ii) is complete.
	\begin{lemma}\label{lem: loc pathwise uniqueness}
		Suppose that the SDE \eqref{eq: def SDE} satisfies strong existence and uniqueness up to explosion. Let \(\psi\) be an \(I\)-valued \(\mathcal{F}_0\)-measurable random variable and let \((V_t)_{t \geq 0}\) be the solution process to \eqref{eq: def SDE} with driver \(W\) and initial value \(\psi\) and let \(\tau\) be a stopping time. Then, all adapted \(I\)-valued continuous processes \((U_t)_{t \geq 0}\) with
		\[
		\dd U_t = \mu(U_t)\1_{\{t \leq \tau\}} \dd t + v (U_t)\1_{\{t \leq \tau\}} \dd W_t, \quad U_0 = \psi, 
		\]
		are indistinguishable from \((V_{t \wedge \tau})_{t \geq 0}\). 
	\end{lemma}
Let \(l_n \searrow l, r_n \nearrow r\) be sequences such that \(l < l_{n+1} < l_n < r_n < r_{n +1} <  r\) and set for a function \(\alpha \colon \mathbb{R}_+ \to [l, r]\)
\[
\tau_n (\alpha) \triangleq \inf(t \in \mathbb{R}_+ \colon \alpha(t) \not \in (l_n, r_n)).
\]
	We conclude from Lemma \ref{lem: loc pathwise uniqueness} and Galmarino's test (see \cite[Lemma III.2.43]{JS}) that for all \(n \in \mathbb{N}\) the SDE 
	\begin{align}\label{eq: stopped sde}
	\dd X^j_t = u(X^j_t, j) \1_{\{t \leq \tau_n(X^j) \}} \dd t + \sigma (X^j_t,j) \1_{\{t \leq \tau_n(X^j)\}} \dd W_t,
	\end{align}
	satisfies weak existence and pathwise uniqueness in the usual sense, see \cite[Definitions 5.3.1, 5.3.2]{KaraShre}.
	Thus, due to \cite[Theorem 18.14]{Kallenberg}, there exists a Borel function \(F^n \colon \mathbb{R} \times C(\mathbb{R}_+, \mathbb{R}) \to C(\mathbb{R}_+, I)\) such that whenever \(X^j\) solves \eqref{eq: stopped sde} with driver \(W= (W_t)_{t \geq 0}\) and (possibly stochastic) initial value \(X^j_0\), then a.s. \(X^j = F^n(X^j_0, W)\).
	
	Lemma \ref{lem: loc pathwise uniqueness} and Galmarino's test yield that a.s.
	\begin{align}\label{eq: gal 1}
	\tau_n (X) = \tau_n(F^n (Y_{\delta \wedge T}, W^\delta)).
	\end{align}
	Because strong existence and uniqueness up to explosion holds for the SDE \eqref{eq: SDE 1}, for a.a. \(\omega \in \Omega\) there exists 
	an \(\F^\delta\)-adapted continuous \([l, r]\)-valued process  \(Y^\omega = (Y^\omega_t)_{t \geq 0}\) such that
	\[
	\dd Y^\omega_t = u(Y^\omega_t, j) \dd t + \sigma (Y^\omega_t, j) \dd W^\delta_t, \quad Y_0^\omega = Y_{\delta(\omega) \wedge T} (\omega) \in I.
	\]
	We stress that the initial value \(Y_{\delta(\omega) \wedge T}(\omega)\) is deterministic.
	Lemma \ref{lem: loc pathwise uniqueness} and Galmarino's test yield that a.s.
	\begin{align}\label{eq: gal2}
	\tau_n (Y^\omega) = \tau_n (F^n(Y_{\delta (\omega) \wedge T} (\omega), W^\delta)).
	\end{align}
	We prove the following lemma after the proof of (ii) is complete.
	\begin{lemma} \label{lem: ind}
		For all \(G \in \mathcal{B}(\mathbb{R}_+)\) we have a.s. \[\p(\zeta - \delta \in G | \mathcal{F}_{\delta \wedge T}, \sigma(W^\delta_t, t \in \mathbb{R}_+)) = - \int_G q_{jj} e^{q_{jj} x} \dd x.\]
	\end{lemma}
Using \eqref{eq: long comp}, the monotone convergence theorem and then \eqref{eq: gal 1},  we obtain that
	\begin{align*}
	0 &= \lim_{n \to \infty} \p (\tau_n (X) \leq \zeta - \delta, \zeta \leq T)
	\\&= \lim_{n \to \infty} \p(\tau_n (F^n (Y_{\delta \wedge T}, W^\delta)) \leq \zeta - \delta, \zeta - \delta + \delta \leq T),
	\intertext{using \cite[Theorem 5.4]{Kallenberg} and Lemma \ref{lem: ind} we further obtain that}
	&= \lim_{n \to \infty} \int_0^T \p(\tau_n(F^n(Y_{\delta \wedge T}, W^\delta)) \leq s, s + \delta \leq T) (- q_{jj}) e^{q_{jj} s} \dd s
	\\&= \lim_{n \to \infty} \int_{0}^T \E^{\p} \big[\p (\tau_n (F^n (Y_{\delta \wedge T}, W^\delta)) \leq s|  \mathcal{F}_{\delta \wedge T}) \1_{\{s + \delta \leq T\}}\big] (- q_{jj}) e^{q_{jj} s}\dd s,
	\intertext{which, due to \cite[Theorem 5.4]{Kallenberg} and the independence of \(W^\delta\) and \(\mathcal{F}_{\delta \wedge T}\), equals}
	&= \lim_{n \to \infty} \int_{0}^T \int_\Omega \p(\tau_n (F^n (Y_{\delta (\omega) \wedge T} (\omega), W^\delta)) \leq s) \1_{\{s + \delta (\omega) \leq T\}} \p( \dd \omega) (-q_{jj})e^{q_{jj} s} \dd s,
	\intertext{and finally, with \eqref{eq: gal2} and the monotone convergence theorem, we obtain}
	&= \lim_{n \to \infty} \int_{0}^T \int_\Omega \p(\tau_n (Y^\omega) \leq s) \1_{\{s + \delta (\omega) \leq T\}} \p( \dd \omega) (- q_{jj}) e^{q_{jj} s} \dd s
	\\&= \int_{0}^T \int_\Omega \p(\tau (Y^\omega) \leq s) \1_{\{s + \delta (\omega) \leq T\}} \p( \dd \omega) (- q_{jj}) e^{q_{jj} s} \dd s.
	\end{align*}
	Due to Feller's test for explosion (see \cite[Theorem 5.5.29]{KaraShre}), \(Y^\omega\) reaches \(l\) or \(r\) in finite time with positive probability. In fact, because \(Y^\omega\) is  regular due to \cite[Proposition 2.2]{mijatovic2012b}, \cite[Theorem 1.1]{bruggeman2016} implies that \(Y^\omega\) even reaches \(l\) or \(r\) arbitrarily fast with positive probability, i.e. \(\p(\tau (Y^\omega) \leq \varepsilon) > 0\) for all \(\varepsilon > 0\). Consequently, the identity 
	\[
	\int_{0}^T \int_\Omega \p(\tau (Y^\omega) \leq s) \1_{\{s + \delta (\omega) \leq T\}} \p( \dd \omega) (- q_{jj}) e^{q_{jj} s} \dd s = 0
	\]
	implies that for \(\llambda\)-a.a. \(s \in (0, T)\) we have \(\p(\delta \leq T - s) = 0\). 
	However, because \(\xi\) is irreducible, we have \(\p(\xi_t= j) > 0\) for all \(t > 0\). This is a contradiction and the proof of (ii) is complete.
\end{proof}
\noindent
\textit{Proof of Lemma \ref{lem: comparison}:}
Define 
\(
\iota \triangleq \zeta \wedge T - \delta \wedge T
\).
Note that for all \(t \in \mathbb{R}_+\)
\begin{align*}
\{\iota \leq t\} = \{\zeta \leq t + \delta \wedge T\} \in \mathcal{F}_{t + \delta \wedge T}, 
\end{align*}
which shows that \(\iota\) is an \(\F^\delta\)-stopping time. 
Moreover, we have for all \(s, t \in \mathbb{R}_+\)
\begin{align*}
\{s \wedge \iota + \delta \wedge T \leq t\} = \big(\{s + \delta &\wedge T \leq t\} \cap \overbrace{\{s + \delta \wedge T \leq \zeta \wedge T\}}^{\in \mathcal{F}_{s + \delta \wedge T}} \big) \\&\cup \big(\{\zeta \wedge T \leq t\} \cap  \underbrace{\{s + \delta \wedge T > \zeta \wedge T\}}_{\in \mathcal{F}_{\zeta \wedge T}}\big) \in \mathcal{F}_t.
\end{align*}
Thus, the random time \(s \wedge \iota + \delta \wedge T\) is an \(\F\)-stopping time.
We deduce from classical rules for time-changed stochastic integrals that a.s. for all \(t \in \mathbb{R}_+\)
\begin{align*}
Y_{t \wedge \iota + \delta \wedge T} &= \phi +  \int_0^{t \wedge \iota + \delta \wedge T} u(Y_s, \xi_{s}) \dd s + \int_0^{t \wedge \iota + \delta \wedge T} \sigma(Y_s, \xi_{s}) \dd W_s
\\&= Y_{\delta \wedge T}+ \int_{0}^{t} u(Y_{s \wedge \iota + \delta \wedge T}, j) \1_{\{s \leq \iota\}} \dd s + \int_0^{t} \sigma(Y_{s \wedge \iota + \delta \wedge T}, j) \1_{\{s \leq \iota\}} \dd W^\delta_s.
\end{align*}	
Because the SDE \eqref{eq: SDE 1} satisfies strong existence and uniqueness up to explosion,
Lemma \ref{lem: loc pathwise uniqueness} implies that a.s. \(Y_{t \wedge \iota + \delta \wedge T} = X_{t \wedge \iota}\) for all \(t \in \mathbb{R}_+\). On \(\{\zeta \leq T\} \subseteq \{\delta \leq T\}\) we have \(\iota = \zeta - \delta\) and the claim follows.
\qed\\\\
\noindent
\textit{Proof of Lemma \ref{lem: loc pathwise uniqueness}:}
Due to localization, we can assume that \(\tau\) is finite. 
By \cite[Proposition V.1.5]{RY} and L\'evy's characterization, the process 
\[
\widehat{W}_t \triangleq W_{t + \tau} - W_\tau,\quad t \in \mathbb{R}_+,
\]  is an \((\mathcal{F}_{t + \tau})_{t \geq 0}\)-Brownian motion.
Due to the strong existence and uniqueness hypothesis, there exists a solution process \(O = (O_t)_{t \geq 0}\) to the SDE
\[
\dd O_t = \mu(O_t) \dd t + v(O_t)\dd \widehat{W}_t, \quad O_0 = U_\tau.
\] 
We set 
\[
Z_t \triangleq \begin{cases} U_t,&t \leq \tau,\\
O_{t - \tau},&t>\tau.
\end{cases}
\]
The process \(Z\) has continuous paths and similar arguments as used in the proof of Theorem \ref{theo: existence Markov} (i) show that it is \(\F\)-adapted.
Let 
\[
\theta^Z_n \triangleq \inf(t \in \mathbb{R}_+ \colon Z_t \not \in (l_n, r_n)).
\]
On \(\{\tau \geq t \wedge \theta^Z_n\}\) we have a.s.
\begin{align*}
Z_{t \wedge \theta^Z_n} 
&= \psi + \int_0^{t \wedge \theta^Z_n} \mu(Z_s)\dd s + \int_0^{t \wedge \theta^Z_n}v (Z_s)\dd W_s.
\end{align*}
Next, we discuss what happens on the set \(\{\tau < t \wedge \theta^Z_n\}\).
Set \[\theta^O_n \triangleq \inf (t \in \mathbb{R}_+ \colon O_t \not \in (l_n, r_n)).\] 
On \(\{\tau < \theta^Z_n\}\) we have a.s. 
\(
\theta^Z_n = \theta^O_n + \tau.
\)
Moreover, note that
\[
t \wedge (\theta^O_n + \tau) - \tau = \begin{cases} \theta^O_n,& \textup{ if }\theta^O_n + \tau \leq t,\\
t - \tau,&\textup{ if } t \leq \theta^O_n + \tau.
\end{cases}
\]
Thus, \(t \wedge (\theta^O_n + \tau) - \tau \leq \theta^O_n\).
Classical rules for time-changed stochastic integrals yield that on \(\{\tau < t \wedge \theta^Z_n\}\) a.s.
\begin{equation*}
\begin{split}
Z_{t \wedge \theta^Z_n}
&= Z_\tau + \int_0^{t \wedge \theta^Z_n   - \tau} \mu(O_s)\dd s + \int_0^{t \wedge \theta^Z_n - \tau} v (O_s)\dd \widehat{W}_s
\\&= Z_\tau + \int_\tau^{t \wedge \theta^Z_n} \mu(O_{s - \tau}) \dd s + \int_{\tau}^{t \wedge \theta^Z_n} v (O_{s - \tau}) \dd W_s
\\&= \psi + \int_0^{t \wedge \theta^Z_n} \mu(Z_s)\dd s + \int_0^{t \wedge \theta^Z_n} v (Z_s)\dd W_s.
\end{split}
\end{equation*}
We conclude that \(Z\) is a solution process of the SDE \eqref{eq: def SDE} with driver \(W\) and initial value \(\psi\).
By the strong existence and uniqueness hypothesis, we conclude that a.s. \(Z = V\). The definition of \(Z\) implies the claim.
\qed
\\\\\noindent
\textit{Proof of Lemma \ref{lem: ind}:}
Denote the Wiener measure with initial value \(x \in \mathbb{R}\) by \(\mathscr{W}_x\) and by \(\mu_j\) the law of a Feller--Markov chain with the same \(Q\)-matrix as \(\xi\) and initial value \(j \in J\). Let \(\mathcal{C}\) be the \(\sigma\)-field on \(C(\mathbb{R}_+, \mathbb{R})\) generated by the coordinate process.
We deduce from Lemma \ref{lem: indep MC BM} in the Appendix, \cite[Proposition 4.1.5, Theorems 4.4.2, 4.4.6]{EK} that \((j, x) \mapsto (\mu_j \otimes \mathscr{W}_x) (F)\) is Borel for every \(F \in \mathcal{D} \otimes \mathcal{C}\) and that the process \((\xi, W)\) is a strong Markov process in the following sense: For all \(F \in \mathcal{D} \otimes \mathcal{C}\) and every a.s. finite stopping time \(\theta\) we have a.s.
\[
\p ((\xi_{\cdot + \theta}, W_{\cdot + \theta}) \in F |\mathcal{F}_\theta) = (\mu_{\xi_\theta} \otimes \mathscr{W}_{W_\theta}) (F).
\]
For all \(A \in \mathcal{D}\) and \(F \in \mathcal{C}\) the strong Markov properties of \(\xi, W\) and \((\xi, W)\) imply that a.s.
\begin{align*}
\p(\xi_{\cdot + \delta \wedge T} \in A&, W_{\cdot + \delta \wedge T} \in F|\mathcal{F}_{\delta \wedge T})\\ &= \mu_{\xi_{\delta \wedge T}}(A)\ \mathscr{W}_{W_{\delta \wedge T}}(F)
\\&= \p(\xi_{\cdot + \delta \wedge T} \in A|\mathcal{F}_{\delta \wedge T})\p(W_{\cdot + \delta \wedge T} \in F|\mathcal{F}_{\delta \wedge T}).
\end{align*}
This implies that \(\sigma(\zeta - \delta)\) and \(\sigma(W^\delta_t, t \in \mathbb{R}_+)\) are independent given \(\mathcal{F}_{\delta \wedge T}\). Now, \cite[Proposition 5.6]{Kallenberg} yields that a.s.
\[
\p(\zeta - \delta \in G |\mathcal{F}_{\delta \wedge T}, \sigma(W^\delta_t, t \in \mathbb{R}_+)) = \p(\zeta - \delta \in G |\mathcal{F}_{\delta \wedge T}).
\]
By the strong Markov property of \(\xi\) and \eqref{eq: exp}, we have for \(F \in \mathcal{F}_\delta\)
\begin{align*}
\p(\zeta - \delta \in G, F) = 
- \int_G q_{jj} e^{q_{jj} x} \dd x\ \p(F).
\end{align*}
The proof is complete.
\qed
\subsection{Local uniqueness}\label{sec: LU}
For the space of continuous functions from \(\mathbb{R}_+\) into \(I\) or \(\mathbb{R}\), we denote by \(\mathcal{C}\) the \(\sigma\)-field generated by the coordinate process. 
Moreover, we denote by \(\mathbf{C}^o \triangleq (\mathcal{C}^o_t)_{t \geq 0}\) the filtration generated by the corresponding coordinate process and by \(\mathbf{C} \triangleq (\mathcal{C}_t)_{t \geq 0}\) its right-continuous version. The image space will be clear from the context. 
Let \[\rho \colon C(\mathbb{R}_+, I) \times D (\mathbb{R}_+, J) \to [0, \infty]\] be a \(\mathbf{C}^o \otimes \mathbf{D}^o\)-stopping time. An example for \(\rho\) is
\[
\tau (\alpha, \omega) \triangleq \inf (t \in \mathbb{R}_+ \colon \alpha (t) \not \in U \text{ or } \omega(t) \not \in V), 
\]
where \(U \subseteq I\) and \(V \subseteq J\) are open:
\begin{lemma}\label{lem: gamma nst}
	\(\tau\) is a \(\mathbf{C}^o \otimes \mathbf{D}^o\)-stopping time.
\end{lemma}
\begin{proof}
	See \cite[Proposition I.4.5]{RY} and \cite[Proposition 2.1.5]{EK}. 
\end{proof}

Let \(u \colon I \times J \to \mathbb{R}\) and \(\sigma \colon I \times J \to \mathbb{R}\backslash \{0\}\) be Borel functions such that \eqref{eq: ES cond} holds, \(\sigma\) satisfies \eqref{eq: MC FT} and the ES conditions for all \(j \in J\) (see Section \ref{sec: MG MS} for this terminology). In other words, we ask that the conditions from part (i) of Theorem \ref{theo: existence Markov} hold.

For \(i = 1, 2\), let \((\Omega^i, \mathcal{F}^i, \F^i, \p^i)\) be a filtered probability space with right-continuous complete filtration \(\F^i = (\mathcal{F}^i_t)_{t \geq 0}\). Let \(W^i = (W^i_t)_{t \geq 0}\) be a one-dimensional Brownian motion, \(\xi^i = (\xi^i_t)_{t \geq 0}\) be a \(J\)-valued irreducible Feller--Markov chain with \(Q\)-matrix \(Q\) and \(\xi^i_0 = j_0 \in J\), and let \(X^i = (X^i_t)_{t \geq 0}\) be an adapted continuous \(I\)-valued process such that
\[
\dd X^i_{t \wedge \rho(Y^i, \xi^i)} = u(X^i_t, \xi^i_t) \1_{\{t \leq \rho (X^i, \xi^i)\}} \dd t + \sigma (X^i_t, \xi^i_t) \1_{\{t \leq \rho (X^i, \xi^i)\}} \dd W^i_t, \quad X^i_0 = y_0 \in I.
\]
It is implicit that the stochastic integrals are well-defined. We stress that \(\xi^1\) and \(\xi^2\) have the same law, because they have the same \(Q\)-matrix, see Example \ref{ex: xi2}.

The main observation of this section is the following:
\begin{theorem}\label{theo: LU} 
	\(\p^1 \circ (X^1_{\cdot \wedge \rho(X^1, \xi^1)}, \xi^1)^{-1} = \p^2 \circ (X^2_{\cdot \wedge \rho(X^2, \xi^2)}, \xi^2)^{-1}\).
\end{theorem}
\begin{proof}
	We follow the Yamada--Watanabe-type idea used in \cite{J80}.
	Define
	\begin{align*}
	\Omega^* &\triangleq C(\mathbb{R}_+, I) \times  C(\mathbb{R}_+, I) \times D(\mathbb{R}_+, J) \times C(\mathbb{R}_+, \mathbb{R}),\\
	\mathcal{F}^* & \triangleq \mathcal{C} \otimes \mathcal{C} \otimes \mathcal{D} \otimes \mathcal{C},
	\end{align*}
	and for \(i =1, 2\)
	\begin{align*}
	Y^i &\colon \Omega^* \to C(\mathbb{R}_+, I), &\hspace{-3.4cm}Y^i(\omega^1, \omega^2, \omega^3, \omega^4) = \omega^i,\\
	Z^1 &\colon \Omega^* \to D(\mathbb{R}_+, J), &\hspace{-3cm}Z^1(\omega^1, \omega^2, \omega^3, \omega^4) = \omega^3,\\
	Z^2 &\colon \Omega^* \to C(\mathbb{R}_+, \mathbb{R}), &\hspace{-3cm}Z^2(\omega^1, \omega^2, \omega^3, \omega^4) = \omega^4.
	\end{align*}
	Denote the Wiener measure by \(\mathscr{W}\) and denote the unique law of \(\xi^i\) by \(\mu\). Due to Lemma \ref{lem: indep MC BM} in the Appendix, we have
	\[
	\p^i \circ (\xi^i, W^i)^{-1} = \mu \otimes \mathscr{W}.
	\]
	When the space of continuous functions is equipped with the local uniform topology it is a Polish spaces and the corresponding Borel \(\sigma\)-fields is generated by the coordinate process. 
	Thus, there exist regular conditional probabilities 
	\[
	Q^i \colon D(\mathbb{R}_+, J) \times C(\mathbb{R}_+, \mathbb{R}) \times \mathcal{C} \to [0, 1]\] such that 
	\begin{align*}
	\p^i (X^i \in \dd \omega^1, \xi^i \in \dd \omega^2, W^i \in \dd \omega^3) &= Q^i(\omega^2, \omega^3, \dd \omega^1) \mu(\dd \omega^2) \mathscr{W} (\dd \omega^3).
	\end{align*}
	We define a probability measure \(\Q\) on \((\Omega^*, \mathcal{F}^*)\) by
	\begin{align*}
	\Q (\dd \omega^1 \times \dd \omega^2 \times \dd \omega^3 \times \dd \omega^4) \triangleq Q^1(\omega^3, \omega^4, \dd \omega^1) Q^2(\omega^3, \omega^4, \dd \omega^2) \mu(\dd \omega^3) \mathscr{W}(\dd \omega^4).
	\end{align*}
	With abuse of notation, denote the \(\Q\)-completion of \(\mathcal{F}^*\) again by \(\mathcal{F}^*\) and denote by \(\mathcal{F}^*_t\) the \(\Q\)-completion of 
	\begin{align*}
	\bigcap_{s  > t} \left(\mathcal{C}_s\otimes \mathcal{C}_s\otimes \mathcal{D}_s \otimes \mathcal{C}_s \right), \quad t \in \mathbb{R}_+.
	\end{align*}
	From now on we consider \((\Omega^*, \mathcal{F}^*, \F^* = (\mathcal{F}^*_t)_{t \geq 0}, \Q)\) as underlying filtered probability space.
	In view of \cite[Propositions 4.6, 5.6]{J80},
	for all \(A \in \mathcal{C}_{t}\) the map 
	\(
	\omega \mapsto Q^i (\omega, A)
	\)
	is measurable w.r.t. the \(\mu \otimes \mathscr{W}\)-completion of \(\bigcap_{s > t} (\mathcal{D}^o_s \otimes \mathcal{C}^o_s)\). In other words, \cite[Hypothesis 10.43]{J79} is satisfied and we deduce from \cite[Lemmata 2.7, 2.9]{J80}, \cite[Proposition 10.46]{J79} and L\'evy's characterization that \(Z^1\) is a Markov chain with \(Q\)-matrix \(Q\), \(Z^2\) is a Brownian motion 
	and
	\begin{align*}
	\dd Y_{t \wedge \rho(Y^i, Z^1)}^i = u(Y_t^i, Z^1_{t})\1_{\{t \leq \rho(Y^i, Z^1)\}} \dd t 
	+ \sigma(Y_t^i, Z^1_{t})\1_{\{t \leq \rho(Y^i, Z^1)\}} \dd Z^2_t,\quad Y^i_0 = y_0.
	\end{align*}
	The proof of the following lemma is given after the proof of Theorem \ref{theo: LU} is complete.	
	\begin{lemma}\label{loc pathwise uniqueness}
		Almost surely \(Y^1_{\cdot \wedge \rho(Y^1, Z^1) \wedge \rho(Y^2, Z^1)} = Y^2_{\cdot \wedge \rho(Y^1, Z^1) \wedge \rho(Y^2, Z^1)}\).
	\end{lemma}
	Due to Galmarino's test, this implies a.s. \(\rho(Y^1, Z^1) = \rho(Y^2, Z^1)\). Thus, a.s. \(Y^1_{\cdot \wedge \rho(Y^1, Z^1)} = Y^2_{\cdot \wedge \rho(Y^2, Z^1)}\) and
	the claim follows from the definition of \(\Q\).
\end{proof}
\noindent
\textit{Proof of Lemma \ref{loc pathwise uniqueness}:}
\textbf{Step 1:} Due to localization, we can assume that \(\rho (Y^1, Z^1) \wedge \rho(Y^2, Z^1)\) is finite.
Recall the following fact (see \cite[Proposition III.3.5]{RY}): If \((Z_t)_{t \geq 0}\) is a Feller--Markov chain for the right-continuous filtration \(\G= (\mathcal{G}_t)_{t \geq 0}\) and \(\gamma\) is a finite \(\G\)-stopping time, then \((Z_{t + \gamma})_{t \geq 0}\) is a Feller--Markov chain for a filtration \((\mathcal{G}_{t + \gamma})_{t \geq 0}\) and both chains have the same \(Q\)-matrix.
Due to Theorem \ref{theo: cdc} (i), for \(i = 1, 2\) there exists a process \((O^i_t)_{t \geq 0}\) defined by 
\[
\dd O^i_t = u(O^i_t, Z^1_{t + \rho(Y^1, Z^1) \wedge \rho(Y^2, Z^1)}) \dd t + \sigma (O^i_t, Z^1_{t + \rho(Y^1, Z^1) \wedge \rho(Y^2, Z^1)}) \dd W^\rho_t,
\]
where 
\[
W^\rho_t \triangleq Z^2_{t + \rho (Y^1, Z^1) \wedge \rho(Y^2, Z^1)} - Z^2_{\rho(Y^1, Z^1) \wedge \rho(Y^2, Z^1)},\quad t \in \mathbb{R}_+,
\]
with initial value \(O^i_0 = Y^i_{\rho(Y^1, Z^1) \wedge \rho(Y^2, Z^1)}\).
Now, set 
\[
V^i_t \triangleq \begin{cases}
Y^i_t,&t \leq \rho(Y^1, Z^1) \wedge \rho(Y^2, Z^1),\\
O^i_{t - \rho(Y^1, Z^1) \wedge \rho(Y^2, Z^1)},&t > \rho(Y^1, Z^1) \wedge \rho(Y^2, Z^1).
\end{cases}
\]
As in the proof of Lemma \ref{lem: loc pathwise uniqueness}, we deduce from classical rules for time-changed stochastic integrals that 
\begin{align}\label{eq: global eq}
\dd V^i_t = u(V^i_t, Z^1_t)\dd t + \sigma(V^i_t, Z^1_t)\dd Z^2_t, \quad V^i_0 = y_0,
\end{align}
i.e. that \(V^1\) and \(V^2\) are global solutions.
Thus, it remains to show a version of pathwise uniqueness for the global equation \eqref{eq: global eq}.

\textbf{Step 2:} We use induction. Let \((\zeta_n)_{n \in \mathbb{N}}\) be the stopping times 
\[
\zeta_1 \triangleq \inf(t \in \mathbb{R}_+ \colon Z^1_t \not = Z^1_0),\quad \zeta_n \triangleq \inf (t \geq \zeta_{n-1} \colon Z^1_{t} \not = Z^1_{\zeta_{n-1}}),\quad n \geq 2.
\]
We stress that \(\zeta_n \nearrow \infty\) as \(n \to \infty\).
Almost surely on \(\{t \leq \zeta_1\}\) we have 
\begin{align*}
V^i_t &= y_0 + \int_0^t u(V^i_s, j_0) \dd s + \int_0^t \sigma(V^i_s, j_0) \dd Z^2_s,\quad i = 1,2.
\end{align*}
Recalling that under the assumptions from Theorem \ref{theo: existence Markov} (i) the SDE \eqref{eq: SDE 1} satisfies strong existence and uniqueness (up to explosion), we deduce from Lemma \ref{lem: loc pathwise uniqueness} that a.s. \(V^1_t = V^2_t\) for all \(t \leq \zeta_1\). 
In case  \(N = 1\), we have \(\zeta_1  = \infty\) and the proof is complete. In the following, we assume that \(N  \geq 2\) in which case a.s. \(\zeta_n < \infty\) for all \(n \in \mathbb{N}\).
Suppose that \(n \in \mathbb{N}\) is such that a.s. \(V^1_t = V^2_t\) for all \(t \leq \zeta_{n}\). Using classical rules for time-changed stochastic integrals, we obtain that a.s. on \(\{t \leq \zeta_{n + 1} - \zeta_n\} \cap \{Z^1_{\zeta_n} = j\}\) 
\begin{align*}
V^i_{t + \zeta_n} &= V^i_{\zeta_n} + \int_{\zeta_n}^{t + \zeta_n} u(V^i_s, j) \dd s + \int_{\zeta_n}^{t + \zeta_n} \sigma(V^i_s, j) \dd Z^2_s
\\&= V^i_{\zeta_n} + \int_{0}^{t} u(V^i_{s + \zeta_n}, j) \dd s + \int_{0}^{t} \sigma(V^i_{s + \zeta_n}, j) \dd W^n_s,
\end{align*}
where 
\[
W^n_t \triangleq Z^2_{t + \zeta_n} - Z^2_{\zeta_n},\quad t \in \mathbb{R}_+.
\]
We conclude again from Lemma \ref{lem: loc pathwise uniqueness} that a.s. \(V^1_{t + \zeta_n} = V^2_{t + \zeta_n}\) for all \(t \leq \zeta_{n + 1} - \zeta_n\). Consequently, a.s. \(V^1_t = V^2_t\) for all \(t \leq \zeta_{n +1}\) and our claim follows.
\qed

\subsection{Proof of Theorem \ref{theo: mart MS}}
\textbf{(i).} Recall that \(J = \{1, \dots, N\}\) with \(1 \leq N \leq \infty\). 
For \(n \in \mathbb{N}\) define 
\[
\tau_n \triangleq \inf (t \in [0, T] \colon S_t \not\in (l_n, r_n) \text{ or } \xi_t \geq n \wedge N).
\]
Because \(c\) is assumed to be bounded on compact subsets of \(I \times J\), Novikov's condition implies that \((\tau_n)_{n \in \mathbb{N}}\) is a localizing sequence for \(Z\). 
We define \(\Q^n\) by the Radon--Nikodym derivative \(\frac{\dd \Q^n}{\dd \p} \triangleq Z_{T \wedge \tau_n}\). By Girsanov's theorem, 
\[
B^n \triangleq W - \int_0^{\cdot \wedge \tau_n} c (S_s, \xi_s)\dd s
\]
is a \(\Q^n\)-Brownian motion such that 
\[
\dd S_{t \wedge \tau_n} = (b (S_t, \xi_t) + c (S_t, \xi_t) \sigma (S_t, \xi_t)) \1_{\{t \leq \tau_n\}} \dd t + \sigma (S_t, \xi_t) \1_{\{t \leq \tau_n\}} \dd B^n_t.
\]
We deduce from Lemma \ref{lem: indep MC BM}, Example \ref{ex: xi2} and Theorem \ref{theo: indp preserving} that under \(\Q^n\) the process \(\xi\) remains a Feller--Markov chain with unchanged \(Q\)-matrix. W.l.o.g. we extend \(W, \xi\) and \(\F\) to the infinite time interval \(\mathbb{R}_+\).
Applying Theorem \ref{theo: existence Markov} with \(u \triangleq b + c \sigma\) yields that on \((\Omega, \mathcal{F}, \F, \p)\) there exists an adapted continuous \(I\)-valued process \(X = (X_t)_{t \geq 0}\) such that 
\[
\dd X_t = (b(X_t, \xi_t) + c(X_t, \xi_t) \sigma (X_t, \xi_t)) \dd t + \sigma (X_t, \xi_t)\dd W_t, \quad X_0 = S_0.
\]
We set 
\[
\rho_n \triangleq \inf (t \in [0, T] \colon X_t \not\in (l_n, r_n) \text{ or } \xi_t \geq n \wedge N).
\]
It follows from Lemma \ref{lem: gamma nst} and Theorem \ref{theo: LU} that 
\[
\p \circ (X_{\cdot \wedge \rho_n}, \xi)^{-1} = \Q^n \circ (S_{\cdot \wedge \tau_n}, \xi)^{-1}.
\]
Consequently, using Galmarino's test, we obtain that 
\[
\lim_{n \to \infty} \Q^n (\tau_n = \infty) = \lim_{n \to \infty} \p (\rho_n = \infty) = 1.
\]
Now, it follows as in the proof of Theorem \ref{theo: mart Ito} that \(Z\) is a martingale.

\textbf{(ii).} This result follows similar as Theorem \ref{theo: general SLM}, where Theorem \ref{theo: existence Markov} has to be used instead of Theorem \ref{theo: 1D Feller p2}. We omit the details.
\qed

\section{Proof of Theorem \ref{theo: indp preserving}}\label{sec: pf theo modi ind}
\textbf{Step 1. }
Let \(g \in A\) and set
\begin{align}\label{eq: Mg}
	M^g_t \triangleq g(\xi_t) - g(\xi_0)- \int_0^t L g(\xi, s) \dd s, \quad t \in [0, T].
\end{align}
Due to the definition of the martingale problem \((A, L, T)\), the process \(M^g\) is a local martingale with localizing sequence \((\rho_n (\xi))_{n \in \mathbb{N}}\).
Thus, the quadratic variation process \([M^g, W]\) is well-defined. Our first step is to show that a.s. \([M^g, W] = 0\).
We explain that \(WM^g\) is a local martingale for the completed right-continuous version of the natural filtration of \(\xi\) and \(W\). Let \(0 \leq s < t \leq T\), \(G \in \sigma (W_r, r \in [0, s]) \triangleq \mathcal{W}_s\) and \(F \in \sigma (\xi_r, r \in [0, s]) \triangleq \mathcal{E}_s\). The independence assumption yields that
\begin{align*}
	\E^\p\big[W_t M^g_{t \wedge \rho_m (\xi)} \1_{G \cap F}\big] &= \E^\p\big[W_t \1_G\big] \E^\p\big[M^g_{t \wedge \rho_m (\xi)} \1_F\big] \\&= \E^\p\big[W_s \1_G\big] \E^\p\big[M^g_{s \wedge \rho_m(\xi)} \1_F\big] \\&= \E^\p\big[W_s M^g_{s \wedge \rho_m(\xi)} \1_{G \cap F}\big].
\end{align*}
By a monotone class argument, we have 
\[
\E^\p\big[W_t M^g_{t \wedge \rho_m(\xi)} \1_B\big] = \E^\p\big[W_s M^g_{s \wedge \rho_m(\xi)} \1_B\big]
\]
for all \(B \in \mathcal{W}_s \vee \mathcal{E}_s\). Due to the downwards theorem (see \cite[Theorem II.51.1]{RW1}), the process \(W M^g_{\cdot \wedge \rho_m(\xi)}\) is a martingale for the completed right-continuous version \(\G \triangleq (\mathcal{G}_t)_{t \in [0, T]}\) of \((\mathcal{W}_t \vee \mathcal{E}_t)_{t \in [0, T]}\). Consequently, because \(\rho_m (\xi)\nearrow \infty\) as \(m \to \infty\), \(W M^g\) is a local \(\G\)-martingale.
By the tower rule, also \(W\) and \(M^g\) are local \(\G\)-martingales.
Integration by parts implies that  
\begin{align*}
	[W,M^g] &= WM^g - \int_0^\cdot W_{s} \dd M^g_s - \int_0^\cdot M^g_{s-} \dd W_s,
\end{align*}
where the stochastic integrals are defined as local \(\G\)-martingales.
Here, we use that \([W, M^g]\) can be defined independently of the filtration. 
We deduce that the process \([W, M^g]\) is a continuous local \(\G\)-martingale of finite variation and hence a.s. \([W, M^g] = 0\).

\textbf{Step 2.} In this step we identify the laws of \(B\) and \(\xi\) under \(\Q\). 
Clearly, \(B\) is a \(\Q\)-Brownian motion due to Girsanov's theorem. 
Next, we show that on \((\Omega, \mathcal{F}, \F, \Q)\) the process \(\xi\) is a solution process for the martingale problem \((A, L, T)\). 
By Step 1 and Girsanov's theorem, the process
\[
M^g - \int_0^t\frac{\dd [Z, M^g]_s}{Z_{s}} = M^g - \int_0^\cdot  \theta_s\dd\hspace{0.05cm} [W, M^g]_s
= M^g
\] 
is a local \(\Q\)-martingale. 
The equivalence \(\Q\sim \p\) implies that \(\Q(\xi_0 = e_0) = 1\) and that \(M^g_{\cdot \wedge \rho_n(\xi)}\) is \(\Q\)-a.s. bounded.
Thus, the claim follows.

\textbf{Step 3.} We prove \(\Q\)-independence of \(B\) and \(\xi\) borrowing an idea from \cite[Theorem 4.10.1]{EK}. 
We define \(C^2_b(\mathbb{R})\) to be the set of all bounded twice continuously differentiable functions \(\mathbb{R}   \to \mathbb{R}\) with bounded first and second derivative. 
Suppose that \(f \in C^2_b(\mathbb{R})\) with \(\inf_{x \in \mathbb{R}} f(x) > 0\) and define 
\[
K^f_t \triangleq f(B_t) \exp \Big( - \frac{1}{2} \int_0^t \frac{f''(B_s)}{f(B_s)} \dd s \Big),\quad t \in [0, T].
\]
By It\^o's formula, we have 
\begin{align*}
	\dd K^f_t &=  \exp \Big(- \frac{1}{2} \int_0^t \frac{f''(B_s)}{f(B_s)} \dd s \Big) \big(\dd f(B_t) - \tfrac{1}{2} f''(B_t) \dd t \big) 
	\\&= \exp \Big(- \frac{1}{2} \int_0^t \frac{f''(B_s)}{f(B_s)} \dd s \Big)f'(B_t)\dd B_t.
\end{align*}
Thus, \(K^f\) is a \(\Q\)-martingale, as it is a bounded local \(\Q\)-martingale.
Recall that the quadratic variation process is not affected by an equivalent change of measure.
By Step 1, \(\Q\)-a.s. \([B, M^g] = 0\).
Due to integration by parts, we obtain that 
\begin{align*}
	\dd K^f_t M^g_{t} &= K^f_t \dd M^g_{t} + M^g_{t -} \dd K^f_t +  \dd [K^f, M^g]_t
	\\&= K^f_t \dd M^g_{t} + M^g_{t-} \dd K^f_t,
\end{align*}
which implies that \(K^f M^g_{\cdot \wedge \rho_m(\xi)}\) is a \(\Q\)-martingale, as it is a bounded local \(\Q\)-martingale.

Let \(\zeta\) be a stopping time such that \(\zeta \leq T\) and set 
\[
\widetilde{\Q} (G) \triangleq \frac{\E^\Q\big[ \1_{G} K^f_\zeta\big]}{\E^\Q\big[K^f_\zeta\big]},\quad G \in \mathcal{F}.
\]
Because \(K^f M^g_{\cdot \wedge \rho_m(\xi)}, K^f\) and \(M^f_{\cdot \wedge \rho_m(\xi)}\) are \(\Q\)-martingales (see also Step 2), the optional stopping theorem implies that for all stopping times \(\psi \leq T\)
\begin{align*}
	\E^{\widetilde{\Q}} \big[ M^g_{\psi \wedge \rho_m(\xi)} \big] &= \frac{\E^\Q \big[M^g_{\psi \wedge \rho_m(\xi)} K^f_\zeta\big]}{\E^\Q \big[K^f_\zeta\big]}
	= 0.
\end{align*}
Consequently, by \cite[Proposition II.1.4]{RY}, \(M^g_{\cdot \wedge \rho_m(\xi)}\) is a \(\widetilde{\Q}\)-martingale. Because \(\widetilde{\Q} \sim \Q\), this implies that on \((\Omega, \mathcal{F}, \F, \widetilde{\Q})\) the process \(\xi\) is a solution process for the martingale problem \((A, L, T)\).
The uniqueness assumption for the martingale problem \((A, L, j_0, T)\) implies that 
\begin{align}\label{eq: conclusion}
	\widetilde{\Q} (\Gamma) = \Q(\Gamma)
\end{align}
for all 
\[
\Gamma \triangleq \big\{\xi_{t_1} \in G_1, \dots, \xi_{t_n} \in G_n\big\},
\]
where \(G_1, \dots, G_n \in \mathcal{B}(J)\) and \(0 \leq t_1 < \dots < t_n \leq T\).
We fix \(\Gamma\) such that \(\Q(\Gamma) > 0\) and define 
\[
\widehat{\Q} (F) \triangleq \frac{\E^\Q \big[ \1_F \1_\Gamma\big]}{\Q(\Gamma)}, \quad F \in \mathcal{F}.
\]
Using the definition of \(\widetilde{\Q}\), \eqref{eq: conclusion}, the fact that \(K^f\) is a \(\Q\)-martingale and the optional stopping theorem, we obtain
\begin{align*}
	\E^{\widehat{\Q}} \big[K^f_\zeta\big] = \frac{\E^\Q \big[K^f_\zeta \1_\Gamma\big]}{\Q(\Gamma)} = \frac{\widetilde{\Q}(\Gamma) \E^\Q \big[K^f_\zeta\big]}{\Q(\Gamma)} = \E^\Q \big[K^f_\zeta\big] = f(0).
\end{align*}
Because \(\zeta\) was arbitrary, we conclude that \(K^f\) is a \(\widehat{\Q}\)-martingale. Furthermore, \(\widehat{\Q}(B_0 = 0) = 1\) follows from the fact that \(B\) is a \(\Q\)-Brownian motion.
Finally, due to \cite[Proposition 4.3.3]{EK}, the process \(B\) is a \(\widehat{\Q}\)-Brownian motion. We conclude that
\[
\widehat{\Q} \big(B_{s_1} \in F_1, \dots, B_{s_k} \in F_k\big) = \Q \big(B_{s_1} \in F_1, \dots, B_{s_k} \in F_k\big),
\]
for all \(F_1, \dots, F_k \in \mathcal{B}(\mathbb{R})\) and \(0 \leq s_1 < \dots < s_k \leq T\).
By the definition of \(\widehat{\Q}\), we have proven that
\begin{align*}
	\Q\big(B_{s_1} &\in F_1, \dots, B_{s_k} \in F_k, \xi_{t_1} \in G_1, \dots, \xi_{t_m} \in G_m\big) \\&= \Q \big(B_{s_1} \in F_1, \dots, B_{s_k} \in F_k\big)  \Q\big(\xi_{t_1} \in G_1, \dots, \xi_{t_m} \in G_n\big),
\end{align*}
which implies that the \(\sigma\)-fields \(\sigma(\xi_t, t \in [0, T])\) and \(\sigma(B_t, t \in [0, T])\) are \(\Q\)-independent. The proof is complete.
\qed

\section{Proof of Theorem \ref{theo: cdc}}\label{sec: pf mg JS}
Because \(\sigma(\xi_t, t \in [0, T])\) and \(\sigma(W_t, t \in [0, T])\) are assumed to be \(\p\)-independent, it follows as in the proof of Theorem \ref{theo: indp preserving} that a.s. \([Z,  W] = 0\). Thus, Girsanov's theorem implies that \(W\) is a \(\Q\)-Brownian motion.

Take \(0 \leq s_1 < \dots < s_m \leq T, 0 \leq t_1 < \dots < t_n \leq T\),\((G_k)_{k \leq m} \subset \mathcal{B}(J)\) and \((F_k)_{k \leq n} \subset \mathcal{B}(\mathbb{R})\), and set 
\begin{align*}
\Gamma_1 &\triangleq \big\{\xi_{s_1} \in G_1, \dots, \xi_{s_m} \in G_m\big\},\\
\Gamma_2 &\triangleq \big\{W_{t_1} \in F_1, \dots, W_{t_n} \in F_n\big\}.
\end{align*}
The \(\p\)-independence of \(\sigma(\xi_t, t \in [0, T])\) and \(\sigma(W_t, t \in [0, T])\) and the uniqueness of the Wiener measure yield that
\begin{align*}
\Q(\Gamma_1 \cap \Gamma_2) &= \E^\p \big[ Z_T \1_{\Gamma_1 \cap \Gamma_2}\big]
\\&= \E^\p \big[Z_T \1_{\Gamma_1} \big] \p(\Gamma_2)
\\&= \Q(\Gamma_1) \Q(\Gamma_2).
\end{align*}
We conclude that \(\sigma(\xi_t, t \in [0, T])\) and \(\sigma(W_t, t \in [0, T])\) are \(\Q\)-independent.

For \(g \in A^*\) we set 
\begin{align*}
	M^g_t &\triangleq g(\xi_t) - g(\xi_0) - \int_0^t L^* g(\xi, s)\dd s, \quad t \in [0, T],\\
	K^{f}_t &\triangleq f(\xi_t)  - f(\xi_0)-  \int_0^t Lf(\xi, s)\dd s, \quad t \in [0, T],\\
	K^{fg}_t &\triangleq f(\xi_t) g(\xi_t) - f(\xi_0) g(\xi_0)- \int_0^t L(fg)(\xi, s)\dd s, \quad t \in [0, T].
\end{align*}
The processes \(K^f\) and \(K^{fg}\) are local \(\p\)-martingales.
We set 
\[
V_t \triangleq  \frac{1}{f(\xi_0)} \exp \Big(- \int_0^t \frac{Lf (\xi, s)}{f(\xi_s)} \dd s \Big),\quad t \in [0, T].
\]
Integration by parts implies that
\begin{align*}
\dd Z_t &= V_t \Big(\dd f(\xi_t) - f(\xi_{t}) \frac{L (\xi, t)}{f(\xi_t)} \dd t\Big) = V_t \dd K^f_t.
\end{align*}
Using again integration by parts and the identity \(L^* g = \frac{1}{f} (L (fg) - g Lf)\) yields
\begin{align*}
\dd Z_t M^g_t &= Z_{t-} \dd M^g_t + M^g_{t-} \dd Z_t + \dd [Z, M^g]_t
\\&= V_t \Big( f(\xi_{t-}) \dd M^g_t + M^g_{t-} \dd K^f_t + \dd [f(\xi), g(\xi)]_t \Big)
\\&= V_t \Big( f(\xi_{t-}) \dd g(\xi_t) - f(\xi_{t-}) L^* g(\xi, t)\dd t + g(\xi_{t-}) \dd f(\xi_t) \\&\qquad - g(\xi_{t-}) L f(\xi, t)\dd t - \Big(g(\xi_0) + \int_0^t L^* g(\xi, s)\dd s \Big)\dd K^f_t + \dd [f(\xi), g(\xi)]_t \Big)
\\&= V_t \Big( \dd \big((fg)(\xi_t)\big) - L (f g)(\xi, t)\dd t - \Big(g(\xi_0) + \int_0^t L^* g(\xi, s)\dd s \Big)\dd K^f_t \Big)
\\&= V_t \Big( \dd K^{fg}_t - \Big(g (\xi_0) + \int_0^t L^* g(\xi, s)\dd s \Big) \dd K^f_t \Big).
\end{align*}
We conclude that \(Z M^g\) is a local \(\p\)-martingale and it follows from \cite[Proposition III.3.8]{JS} that \(M^g\) is a local \(\Q\)-martingale. 
Due to the equivalence \(\Q \sim \p\), we conclude that on \((\Omega, \mathcal{F}, \F, \Q)\) the process \(\xi\) is a solution process to the martingale problem \((A^*, L^*, j_0, T)\).
 \qed
\section{Proof of Theorem \ref{prop: mg f}}
Let \((X_t)_{t \geq 0}\) be the coordinate process on \(D (\mathbb{R}_+, J)\) and denote
\[
M^f_t \triangleq \frac{f(X_t)}{f(j_0)} \exp \Big(- \int_0^t \frac{L f(X, s)}{f(X_{s})} \dd s\Big), \quad t \in [0, T].
\]
Define by \(\mu \triangleq \p \circ \xi^{-1}\) a Borel probability measure on \(D (\mathbb{R}_+, J)\).
We have to show that \[
\E^\mu\big[M^f_{T}\big] = 1.\]
It follows from \cite[Lemma 2.9]{J79} that \(M^f\) is a local \(\mu\)-martingale with localizing sequence \((\rho_n)_{n \in \mathbb{N}}\).
For all \(n \in \mathbb{N}\), define a Borel probability measure \(\mu_n\) on \(D (\mathbb{R}_+, J)\) via the Radon--Nikodym derivative \[\frac{\dd \mu_n}{\dd \mu} = M^f_{T \wedge \rho_n}.\]
The following lemma is proven after the proof of Theorem \ref{prop: mg f} is complete.
\begin{lemma}\label{lem:luJS}
	Let \(\mu^*\) be the unique law of a solution process to the martingale problem \((A^*, L^*, j_0, \infty)\).
	For all \(n \in \mathbb{N}\) we have \(\mu_n = \mu^*\) on \(\mathcal{D}^o_{T \wedge \rho_n}\).
\end{lemma}
Recalling that \(\{\rho_n > T\} \in \mathcal{D}^o_{T \wedge \rho_n}\), Lemma \ref{lem:luJS} implies that
\[
\E^\mu \big[M^f_T\big] = \lim_{n \to \infty} \E^\mu \big[M^f_{T \wedge \rho_n} \1_{\{\rho_n > T\}}\big] = \lim_{n \to \infty} \mu^*(\rho_n > T) = 1.
\]
This completes the proof. \qed
\\\\
\noindent
\textit{Proof of Lemma \ref{lem:luJS}:}
We adapt the proof of \cite[Theorem III.2.40]{JS}. 
To simplify our notation, we set \(T \wedge \rho_n\triangleq\rho\).
We denote by \(\mu_j\) the unique law of a solution process to the martingale problem \((A^*, L^*, j, \infty)\). 

\textbf{Step 1.} We show that \(j \mapsto \mu_j (G)\) is Borel for all \(G \in \mathcal{D}\), following the strategy outlined in \cite[Exercise 6.7.4]{SV}. 
Recall that we assume that \(A^*\) contains a countable determining set \(\widetilde{A}\). 
Let \(\mathcal{P}\) be the space of Borel probability measures on \(D(\mathbb{R}_+, J)\) equipped with the topology of convergence in distribution.
Note that a \(\D^o\)-adapted process is a \(\D\)-martingale if and only if it is a \(\D^o\)-martingale. The implication \(\Rightarrow\) follows from the downward theorem (see \cite[Theorem II.51.1]{RW1}) and the implication \(\Leftarrow\) follows from the tower rule.  For \(g \in \widetilde{A}\) set
\[
K^g_{t} \triangleq g(X_{t}) - g(X_0) - \int_0^{t} K g(X_{s}) \dd s, \quad t \in \mathbb{R}_+.\]
Define \(\mathcal{I}\) to be the set of all \(\nu \in \mathcal{P}\) such that \(\nu \circ X_0^{-1} = \delta_j\) for some \(j \in J\). Moreover, let \(\mathcal{M}\) be the set of all \(\nu \in \mathcal{P}\) such that
\[
\E^\nu \big[ (K^g_{t \wedge \rho_m} - K^g_{s \wedge \rho_m})\1_G\big] = 0, 
\]
for all \(g \in \widetilde{A}\), all rational \(s < t, m \in \mathbb{N}\) and \(G\) in a countable determining class of \(\mathcal{D}^o_s\). 
By the uniqueness assumption, \(\{\mu_j, j \in J\} = \mathcal{I} \cap \mathcal{M}\). Because the set \(\{\delta_j, j \in J\}\) is Borel due to \cite[Theorem 8.3.7]{cohn13} and \(\nu \mapsto \nu \circ X_0^{-1}\) is continuous, \(\mathcal{I}\) is Borel. 
The set \(\mathcal{M}\) is Borel due to \cite[Theorem 15.13]{aliprantis2013infinite}.
We conclude that \(\{\mu_j, j \in J\}\) is Borel. 
Let \(\Phi \colon \{\mu_j, j \in J\} \to J\) be defined by \(\Phi(\mu_j) = j\) for all \(j \in J\). We note that \(\Phi\) is a continuous injection. Thus, the inverse map \(\Phi^{-1}\) is Borel due to Kuratovski's theorem (\cite[Proposition 8.3.5]{cohn13}). This means that \(j \mapsto \mu_j(G)\) is Borel for all \(G \in \mathcal{D}\).

\textbf{Step 2.} Because \(\mu_n \sim \mu\), we have \(\mu_n(X_0 = j_0) = 1\).
As in the proof of Theorem \ref{theo: cdc}, we see that for all \(g \in A^*\) the process \(K^g_{\cdot \wedge \rho}\)
is a \(\mu_n\)-martingale.

\textbf{Step 3.} For every \(t \in \mathbb{R}_+\) we denote by \(\theta_t \colon D(\mathbb{R}_+, J) \to D(\mathbb{R}_+, J)\) the shift operator given by \(\theta_t \omega (s) = \omega(t + s)\).
Recalling that \(\rho\) is bounded, we deduce from \cite[Lemma III.2.44]{JS} that \[\mathcal{D}^o_\rho \vee \theta^{-1}_\rho (\mathcal{D}) = \mathcal{D}.\] 
Hence, we can associate to each \(G \in \mathcal{D}\) a (not necessarily unique) \(G' \in \mathcal{D}^o_{\rho} \otimes \mathcal{D}\) such that 
\[
G = \big\{\omega \in D \colon (\omega, \theta_{\rho(\omega)}\omega) \in G'\big\}.
\]
We define 
\[
\nu (G) \triangleq \int \mu_n(\dd \omega) \mu_{\omega(\rho(\omega))} (\dd \omega^*) \1_{G'} (\omega, \omega^*).
\]
It follows from \cite[Lemma III.2.47]{JS} that \(\nu\) is a probability measure, i.e. that \(\nu\) is defined unambiguously.
Our goal is to show that \(\nu\) solves the martingale problem \((A^*, L^*, j_0, \infty)\).
Providing an intuition, \(\nu\) is the law of 
\[
\begin{cases} X^1_t,& t < \rho (X^1),\\
X^2_{t - \rho (X^1)}, &t \geq \rho (X^1), 
\end{cases}
\]
in case \(X^1\) is sampled according to \(\mu_n\) and \(X^2\) is sampled according to \(\mu_{j}\) with \(j = X^1_{\rho (X^1)}\). In other words, we extend \(\mu_n\) to a solution of the global martingale problem.
For \(G \in \mathcal{D}^o_0\) we can choose \(G' = G \times D (\mathbb{R}_+, J)\). Consequently, 
\[
\nu(X_0 = j_0) = \mu_n(X_0 = j_0) = 1. 
\]
Let \(\psi\) be a bounded \(\D^o\)-stopping time and fix \(m \in \mathbb{N}\).
For \(\omega, \alpha \in D(\mathbb{R}_+, J)\) and \(t \in \mathbb{R}_+\) we set
\[
z(\omega, \alpha) (t) \triangleq \begin{cases} \omega(t),&t < \rho(\omega),\\
\alpha(t - \rho(\omega)),&t \geq \rho(\omega),\end{cases}
\]
and
\[
V(\omega, \alpha) \triangleq \begin{cases} \big((\psi \wedge \rho_m) \vee \rho - \rho\big) (z(\omega, \alpha)), &\alpha (0) = \omega (\rho(\omega)),\\
0,&\text{otherwise}.\end{cases}
\]
Due to \cite[Theorem IV.103]{DellacherieMeyer78} the map
\(V\) is \(\mathcal{D}^o_\rho \otimes \mathcal{D}\)-measurable such that \(V(\omega, \cdot)\) is a \(\D^o\)-stopping time for all \(\omega \in D(\mathbb{R}_+, J)\). Furthermore, it is evident from the definition that
\[
(\psi \wedge \rho_m) (\omega) \vee \rho(\omega) = \rho(\omega) + V(\omega, \theta_{\rho(\omega)} \omega)
\]
for \(\omega \in D (\mathbb{R}_+, J)\). 
For all \(\omega \in \{\rho< \psi \wedge \rho_m\} \in \mathcal{D}_\rho^o\) and \(\alpha \in D(\mathbb{R}_+, J)\) with \(\alpha (0) = \omega(\rho(\omega))\) we have
\(
V(\omega, \alpha) 
\leq \rho_m(\alpha).
\)
Note further that for \(\omega \in \{\rho < \psi \wedge \rho_m\}\)
\begin{align*}
K^g_{V(\omega, \theta_{\rho(\omega)} \omega)} (\theta_{\rho(\omega)}\omega) &= K^g_{(\psi \wedge \rho_m)(\omega) - \rho(\omega)} (\theta_{\rho(\omega)} \omega )  \\&= g(\omega((\psi \wedge \rho_m)(\omega))) - g(\omega(\rho(\omega))) - \int_{\rho(\omega)}^{(\psi \wedge \rho_m)(\omega)} Kg (\omega(s)) \dd s
\\&= K^g_{(\psi \wedge \rho_m)(\omega)}(\omega) - K^g_{\rho(\omega)}(\omega).
\end{align*}
Because \(K^g_{\cdot \wedge \rho}\) is a \(\mu_n\)-martingale, we have
\[
\E^\nu \big[ K^g_{\rho \wedge \psi \wedge \rho_m}\big] = \E^{\mu_n} \big[K^g_{\rho \wedge \psi\wedge \rho_m}\big] = 0,
\]
due to the optional stopping theorem. 
Therefore, we obtain
\begin{align*}
\E^\nu \big[ K^g_{\psi \wedge \rho_m} \big] &= \E^\nu \big[ K^g_{\psi \wedge \rho_m} - K^g_{\rho \wedge \psi \wedge \rho_m}\big] 
\\&= \E^\nu \big[ \big(K^g_{\psi \wedge \rho_m} - K^g_{\rho}\big) \1_{\{\rho < \psi \wedge \rho_m\}}\big]
\\&= \E^\nu\big[K^g_{V (\cdot, \theta_\rho)} ( \theta_\rho ) \1_{\{\rho < \psi \wedge \rho_m\}}\big]
\\&=\int \mu_n(\dd \omega) \E^{\mu_{\omega(\rho(\omega))}} \big[ K^g_{V (\omega, \cdot) \wedge \rho_m}\big] \1_{\{\rho(\omega) < (\psi \wedge \rho_m)(\omega)\}} = 0,
\end{align*}
again due to the optional stopping theorem (recall that \(V(\omega, \cdot)\) is bounded and that \(K^g_{\cdot \wedge \rho_m}\) is a \(\mu_j\)-martingale for all \(j \in J\)). We conclude from \cite[Proposition II.1.4]{RY} that \(K^g_{\cdot \wedge \rho_m}\) is a \(\nu\)-martingale
and hence that under \(\nu\) the coordinate process \((X_t)_{t \geq 0}\) solves the martingale problem \((A^*, L^*, j_0, \infty)\). The uniqueness assumption implies that \(\nu = \mu^*\). Because also for \(G \in \mathcal{D}^o_\rho\) we can choose \(G' = G \times D(\mathbb{R}_+, J)\), we obtain that 
\(
\mu^* (G) = \nu(G) = \mu_n(G). 
\)
The proof is complete.
\qed

\appendix
\section{Independence of Markov chains and Brownian motion}
Let \(0 <T \leq \infty\) be a time horizon and let \((\Omega, \mathcal{F}, \F, \p)\) be a complete probability space with right-continuous and complete filtration \(\F = (\mathcal{F}_t)_{t \in [0, T]}\), which supports a one-dimensional Brownian motion \(W = (W_t)_{t \in [0, T]}\) and a Feller--Markov chain \(\xi =  (\xi_t)_{t \in [0, T]}\).  We suppose that the initial value \(\xi_0\) of \(\xi\) is deterministic. Recall our convention that \(W\) is a Brownian motion for \(\F\) and that \(\xi\) is a Markov chain for \(\F\).
\begin{lemma}\label{lem: indep MC BM}
	\(\sigma  (W_t, t \in [0, T])\) and \(\sigma (\xi_t, t \in [0, T])\) are independent. 
\end{lemma}
\begin{proof}
Let \(f \in C_0(J)\) be such that \(Q f \in C_0(J)\), where \(Q\) is the \(Q\)-matrix of \(\xi\). We set 
\[
M \triangleq f(\xi) - f(\xi_0) - \int_0^\cdot Q f(\xi_s) \dd s.
\]
Let \(g \in C^2_b(\mathbb{R})\) such that \(\inf_{x \in \mathbb{R}} g(x) > 0\) and set 
\[
K \triangleq g(W) \exp  \Big(- \frac{1}{2} \int_0^\cdot \frac{g''(W_s)}{g(W_s)} \dd s \Big).
\]
It\^o's formula yields that \(K\) is a martingale. The process \(M\) is a martingale by Dynkin's formula. Thus, integration by parts yields that
\begin{align}\label{eq: iden ibp}
\dd M_t K_t = M_{t} \dd K_t + K_t \dd M_t + \dd [M, K]_t.
\end{align}
Because \(\xi\) has only finitely many jumps in a finite time interval, the process \(M\) is of finite variation over finite time intervals. Thus, since \(K\) has continuous paths, it follows that a.s. \([M, K] =0\). In view of \eqref{eq: iden ibp}, we conclude that \(MK\) is a martingale. Now, recalling Example \ref{ex: xi2}, we can argue as in Step 3 of the proof of Theorem \ref{theo: indp preserving} to obtain the claimed independence. \end{proof}

\bibliographystyle{agsm}

\end{document}